\documentclass[10pt,a4paper]{article}

\usepackage{amsmath,amsgen,latexsym}
\usepackage{amstext,amssymb,amsfonts,latexsym}
\usepackage{theorem}
\usepackage{pifont}
\usepackage{graphics,epsfig}
\usepackage{graphicx}

 \newcommand{\bs}{\bigskip}
 \newcommand{\ms}{\medskip}
 \newcommand{\n}{\noindent}
 \newcommand{\s}{\smallskip}
 \newcommand{\hs}[1]{\hspace*{ #1 mm}}
 \newcommand{\vs}[1]{\vspace*{ #1 mm}}

 \newcommand{\setempty}{\varnothing}
 
 \newcommand{\nat}{\mathbb{N}}
 
 \newcommand{\integer}{\mathbb{Z}}
 
 \newcommand{\complex}{\mathbb{C}}

 \newcommand{\co}{\mathrm{co}\mbox{-}}

 \newcommand{\AAA}{{\cal A}}
 \newcommand{\BB}{{\cal B}}
 
 \newcommand{\CC}{{\cal C}}
 \newcommand{\FF}{{\cal F}}

 \newcommand{\HH}{{\cal H}}

 \newcommand{\LL}{{\cal L}}
 \newcommand{\NN}{{\cal N}}
 
 \newcommand{\MM}{{\cal M}}
 \newcommand{\SSS}{{\cal S}}

 \newcommand{\dl}{\mathrm{L}}
 \newcommand{\nl}{\mathrm{NL}}

 \newcommand{\pp}{\mathrm{PP}}

 \newcommand{\poly}{\mathrm{poly}}

\theoremstyle{plain}
\theoremheaderfont{\bfseries}
 \newtheorem{theorem}{Theorem}[section]
 \newtheorem{lemma}[theorem]{Lemma}
\newtheorem{proposition}[theorem]{{\bf Proposition}}
 \newtheorem{corollary}[theorem]{Corollary}

 {\theorembodyfont{\rmfamily}
  }
 {\theorembodyfont{\rmfamily} }
 {\theorembodyfont{\rmfamily} }

 \newtheorem{claim}{Claim}

 \newenvironment{proofof}[1]{\vspace*{5mm} \par \noindent
         {\bf Proof of #1.\hs{2}}}{\hfill$\Box$ \vspace*{3mm}}

 \newenvironment{proof}{\par \noindent
            {\bf Proof. \hs{2}}}{\hfill$\Box$ \vspace*{3mm}}

 \newcommand{\ceilings}[1]{\lceil #1 \rceil}
 
 \newcommand{\pair}[1]{\langle #1 \rangle}

 \newcommand{\qubit}[1]{| #1 \rangle}
 \newcommand{\bra}[1]{\langle #1 |}
 
 \newcommand{\measure}[2]{\langle #1 | #2 \rangle}

\newcommand{\ignore}[1]{}

\newcommand{\cent}{|\!\! \mathrm{c}}
\newcommand{\dollar}{\$}

 \newcommand{\psublin}{\mathrm{PsubLIN}}
 \newcommand{\dstcon}{\mathrm{DSTCON}}

 \newcommand{\oned}{1\mathrm{D}}
 \newcommand{\oneq}{1\mathrm{Q}}
 \newcommand{\onebq}{1\mathrm{BQ}}
 \newcommand{\onep}{1\mathrm{P}}
 \newcommand{\onebp}{1\mathrm{BP}}
 \newcommand{\onen}{1\mathrm{N}}
 \newcommand{\twod}{2\mathrm{D}}
 \newcommand{\twon}{2\mathrm{N}}
 \newcommand{\twobq}{2\mathrm{BQ}}
 \newcommand{\twoq}{2\mathrm{Q}}
 \newcommand{\bql}{\mathrm{BQL}}
 
 \newcommand{\bpl}{\mathrm{BPL}}
 \newcommand{\qpoly}{\mathrm{Qpoly}}
 \newcommand{\twobp}{2\mathrm{BP}}
 \newcommand{\twop}{2\mathrm{P}}
 
 \newcommand{\para}{\mathrm{para}\mbox{-}}
 \newcommand{\twosbq}{2\mathrm{sBQ}}
 \newcommand{\onenq}{1\mathrm{NQ}}
 
 \newcommand{\phsp}{\mathrm{PHSP}}

 \newcommand{\NEQ}{{\cal NEQ}}

 \newcommand{\pt}{\mathrm{ptime}\mbox{-}}
 \newcommand{\density}[2]{| #1 \rangle\!\langle #2 |}

\begin{document}

\begin{center}
{\Large {\bf Nonuniform Families of Polynomial-Size Quantum Finite Automata and Quantum Logarithmic-Space Computation with Polynomial-Size Advice}}\footnote{A preliminary version appeared in the Proceedings of the 13th International Conference on Language and Automata Theory and Applications (LATA 2019), Saint Petersburg, Russia, March 26--29, 2019,
Lecture Notes in Computer Science, Springer, vol. 11417, pp. 134--145, 2019. This current paper corrects and extends the preliminary version.} \bs\ms\\

{\sc Tomoyuki Yamakami}\footnote{Present Affiliation: Faculty of Engineering, University of Fukui, 3-9-1 Bunkyo, Fukui 910-8507, Japan} \bs\\
\end{center}

\begin{abstract}
The state complexity of a finite(-state) automaton intuitively measures the size of the description of the automaton. Sakoda and Sipser [STOC 1972, pp. 275--286]   were concerned with nonuniform families of finite automata and they discussed the behaviors of the nonuniform complexity classes defined by such families of finite automata having polynomial-size state complexity.
In a similar fashion, we introduce nonuniform state complexity classes using nonuniform families of quantum finite automata empowered by the flexible use of garbage tapes.
We first present general inclusion and separation relationships among nonuniform state complexity classes of various one-way finite automata, including deterministic, nondeterministic, probabilistic, and quantum  finite automata having polynomially many inner states.
For two-way quantum finite automata equipped with flexible garbage tapes,
we show a close relationship between the nonuniform state complexity of the family of such polynomial-size quantum finite automata and the parameterized complexity class induced by logarithmic-space quantum computation assisted by  polynomial-size advice.
We further establish a direct connection between space-bounded quantum computation with quantum advice and quantum finite automata whose transitions are dictated by superpositions of transition tables.

\s
\n{\bf keywords:}
 nonuniform state complexity,
 quantum finite automata,
 flexible garbage tape,
 quantum Turing machine,
 quantum advice,
 super quantum finite automata
\end{abstract}


\sloppy
\section{Prelude: Quick Overview}\label{sec:introduction}

This exposition reports a collection of fundamental results obtained by our initial study on the state complexity of nonuniform families of various finite automata. Such a complexity measure is briefly referred to as the \emph{nonuniform state complexity} throughout this exposition.

\subsection{Nonuniform State Complexity of Finite Automata Families}\label{sec:nonuniform-sc}

Each finite(-state) automaton is completely described in terms of a set of transitions of its  inner states because there is no memory device, such as a stack, a work tape, etc., other than a read-only input tape.
The number of such inner states is thus crucial to measure the ``descriptional size'' of the automaton in question and it works as a general complexity measure, known as the \emph{state complexity} of the automaton. This state complexity therefore naturally serves as a clear indicator to scale the computational power of the automaton.
Instead of taking a single automaton, in this exposition, we consider a ``family'' of finite automata in a way similar to taking a family of Boolean circuits. Such a family of finite automata may or may not be generated individually by a certain fixed deterministic procedure in a certain uniform manner.
Unlike Boolean circuits, nevertheless, inputs of automata are \emph{not limited} to fixed lengths and this situation provides an additional consideration to the simulation of automata in general. In this exposition, for clarity, we intend to use the term ``uniform state complexity'' to refer  to the state complexity of a uniform family of finite automata. Opposed to this uniform state complexity, we attempt to study its ``nonuniform'' counterpart, which we intend to call \emph{nonuniform state complexity}.
This nonuniform complexity measure has turned out to be closely related to a nonuniform model of Turing-machine computation.

In the past literature, nonuniform state complexity has played various roles in automata theory. An early discussion that attempted to relate certain state complexity issues to the collapses of known space-bounded complexity classes dates back to late 1970s. Sakoda and Sipser \cite{SS78}, following Berman and Lingas \cite{BL77}, argued the state complexity of transforming one family of 2-way nondeterministic finite automata (or 2nfa's, for short) into another family of 2-way deterministic finite automata (or 2dfa's).
From their work, we have come to know that the state complexity of a family of automata is closely related to the work-tape space usage of a Turing machine.
In this line of study, after a long recess, Kapoutsis \cite{Kap14} and Kapoutsis and Pighizzini \cite{KP15} lately revitalized a discussion on the relationships between logarithmic-space (or log-space, for short) complexity classes and state complexity classes in connection to the $\dl=\nl$ question (in fact, the $\nl\subseteq \dl/\poly$ question, where $\dl/\poly$ is the nonuniform analogue of $\dl$).

Taking a complexity-theoretic approach, Kapoutsis \cite{Kap09,Kap12} earlier studied relationships among the nonuniform state complexity classes $\oned$ (one-way deterministic), $\onen$ (one-way nondeterministic), $\twod$ (two-way deterministic), and $\twon$ (two-way nondeterministic). These classes constitute families of ``promise'' decision problems, each family of which is solved by an appropriately chosen nonuniform family of deterministic or nondeterministic finite automata of polynomially many inner states. In contrast, two more nonuniform state complexity classes
$2^{\oned}$ and $2^{\onen}$ were introduced by the use of finite automata  having exponentially many inner states (see Section \ref{sec:complexity-class} for their precise definitions).
There are only a few known separations on those nonuniform state complexity classes. For instance, Kapoutsis \cite{Kap09,Kap12} demonstrated that  $\oned\subsetneqq \onen \subsetneqq 2^{\oned}$,  $\onen\neq \co\onen$, and  $\oned\subsetneqq \twod\subseteq \twon \subsetneqq 2^{\oned}$.
Along the same line of study, Yamakami \cite{Yam18} recently gave a precise  characterization of the polynomial-time sub-linear-space parameterized   complexity class $\psublin$ and an $\nl$-complete problem $3\dstcon$ parameterized by the number of vertices of an input graph (which is generally referred to as a \emph{size parameter}) in terms of the uniform state complexities of restricted 2nfa's and narrow 2-way alternating finite automata (abbreviated as 2afa's).

An important discovery of \cite{Yam18}, nevertheless, is the fact that a nonuniform family of promise decision problems is more closely related to parameterized decision problems than to ``standard'' decision problems (whose complexities are measured by the ``binary'' encoding size of inputs).
A decision problem (identified freely with a language) $L$ over alphabet $\Sigma$ together with a reasonable size parameter $m$, which is a map from $\Sigma^*$ to the set $\nat$ of all natural numbers, forms a \emph{parameterized decision problem} $(L,m)$ \cite{Yam17a}. We can naturally translate such a parameterized decision problem $(L,m)$ into a uniform family $\{(L_n^{(+)},L_n^{(-)})\}_{n\in\nat}$ of promise decision problems satisfying $L_n^{(+)}\cup L_n^{(-)}=\Sigma_n$, where $\Sigma_n$ is a set of input strings having ``size'' $n$, and  we can also translate $\{(L_n^{(+)},L_n^{(-)})\}_{n\in\nat}$ back into another parameterized decision problem $(K,m')$, which is ``almost'' the same as $(L,m)$.
These translations between parameterized decision problems and families of promise decision problems play an essential role in this exposition.
For notational readability, following \cite{Yam19a}, we use the special prefix ``para-'' and write, for example, $\para\dl$ and $\para\nl$ to denote respectively the parameterized analogues of $\dl$ and $\nl$. See Section \ref{sec:parameter-promise} for the precise definitions and descriptions of the aforementioned concepts.

After the study of state complexity classes was initiated in \cite{SS78}, a   further elaboration has been long anticipated; however, there has been little research on how to expand the scope of these classes. Our primary purpose of this exposition is to enrich the world of nonuniform state complexity classes and to lead our study to a whole new direction.

\subsection{An Extension to Quantum Finite Automata}

A surge of papers has already dealt with ``standard'' state complexity issues in an emerging field of \emph{quantum finite automata} \cite{AN09,BMP14,YS10,ZQG+13,ZGQ14}.
We intend to further expand the scope of \emph{nonuniform state complexity theory} to this new field. The behaviors of quantum finite automata, viewed as a natural extension of probabilistic finite automata, are fundamentally governed by \emph{quantum physics}. Moore and Crutchfield \cite{MC00} and Kondacs and Watrous \cite{KW97} modeled the quantization of finite automata in two quite different ways. Lately, these definitions have been considered insufficient for implementation and advantages over classical finite automata and, for this reason,  numerous generalizations have been proposed (see, e.g., a survey \cite{AY15} for references therein).
Here, we intend to use two distinct models:  \emph{measure-many 1-way\footnote{We use this term ``1-way'' \emph{in a strict sense} that a tape head always moves to the right and is not allowed at any time to stay still on the same cell. This term is also called ``real time'' in certain literature.} quantum finite automata with garbage tapes} (or 1qfa's, for short) and \emph{measure-many 2-way quantum finite automata with garbage tapes} (or 2qfa's), where garbage tapes are write-once\footnote{A tape is \emph{write-once} if its tape head never moves to the left and, whenever it writes a non-blank symbol, it must move to the right.} tapes used to discard unwanted information, which is thought to be released into an external environment surrounding the target quantum finite automata. For  early use of extra tape tracks to discard the unnecessary information, see, e.g.,  \cite[Section 5.2]{Yam14b}.
With the ``flexible'' use of such a garbage tape, a quantum finite automaton can freely choose the timing of dumping unwanted information onto the garbage tape. This flexible-garbage-tape model is simple to describe with no additional use of mixed states, superoperators, classical inner states, etc., and even the inflexible use of garbage tapes (called \emph{rigid  garbage tapes}) makes 1qfa's as powerful as generalized models cited in \cite{Hir10,YS11} (see Section \ref{sec:machine-models}).

\subsection{Overview of Main Contributions}\label{sec:main-contribution}

In this exposition, slightly deviating from the past major literature but consistent with \cite{Yam18}, we \emph{do not} allow underlying input alphabets to vary over promise decision problems in  each fixed nonuniform family. Refer to Section \ref{sec:complexity-class} for a precise definition.
In analogy to Kapoutsis's $\oned$ and $\twod$, we introduce their probabilistic and quantum variants in the following manner.
The nonuniform state complexity class $\oneq$ (one-way unbounded-error quantum) is the collection of nonuniform families $\{(L_n^{(+)},L_n^{(-)})\}_{n\in\nat}$, where each promise decision problem $(L_n^{(+)},L_n^{(-)})$ can be solved with unbounded-error probability by a certain 1qfa $M_n$ having  polynomially many inner states and a polynomially-bounded garbage alphabet used for a flexible garbage tape.
If we relax the unbounded-error requirement to the bounded-error requirement
(i.e., error probability is bounded from above by a certain constant in $[0,1/2)$), we write $\onebq$ (one-way bounded-error quantum) in place of $\oneq$.
The nonuniform state complexity classes  $\twobq$ (two-bounded-error quantum) and $\twoq$ (two-way unbounded-error quantum) are introduced respectively in a way similar to $\onebq$ and $\oneq$ but using bounded-error and unbounded-error 2qfa's instead of bounded-error and unbounded-error 1qfa's.
When nondeterministic quantum computation is used instead,
we further obtain $\onenq$ (one-way nondeterministic quantum).
Moreover, if we simply substitute probabilistic finite automata for quantum finite automata, then we also obtain $\onebp$ (one-way bounded-error probabilistic), $\onep$ (one-way unbounded-error probabilistic), $\twobp$ (two-way bounded-error probabilistic), and $\twop$ (two-way unbounded-error probabilistic) from $\onebq$,  $\oneq$,  $\twobq$, and $\twoq$, respectively.
Probabilistic finite automata are allowed to take \emph{arbitrary} real transition probabilities.
For $\twobp$, for example, the classical results of Dwork and Stockmeyer \cite{DS90} derive that $\twobp \subseteq 2^{\oned}$  and $\twobp \nsubseteq\twon$ (Lemma \ref{Dwork-Stockmeyer}).

The first part of our main result is nicely summarized in Figure \ref{fig:relationship}. The newly added inclusion and separation relationships in the figure will be proven in Theorems \ref{main-theorem} and Lemmas \ref{Dwork-Stockmeyer} and \ref{first-claim}.
With the use of rigid garbage tapes, we further define $\oneq^{(+)}$  similarly to $\oneq$.
Lemma \ref{rigid-garbage-tape} additionally proves two relationships regarding $\oneq^{(+)}$ in Figure \ref{fig:relationship}.
\begin{figure}[t]
\centering
\includegraphics*[width=8.5cm]{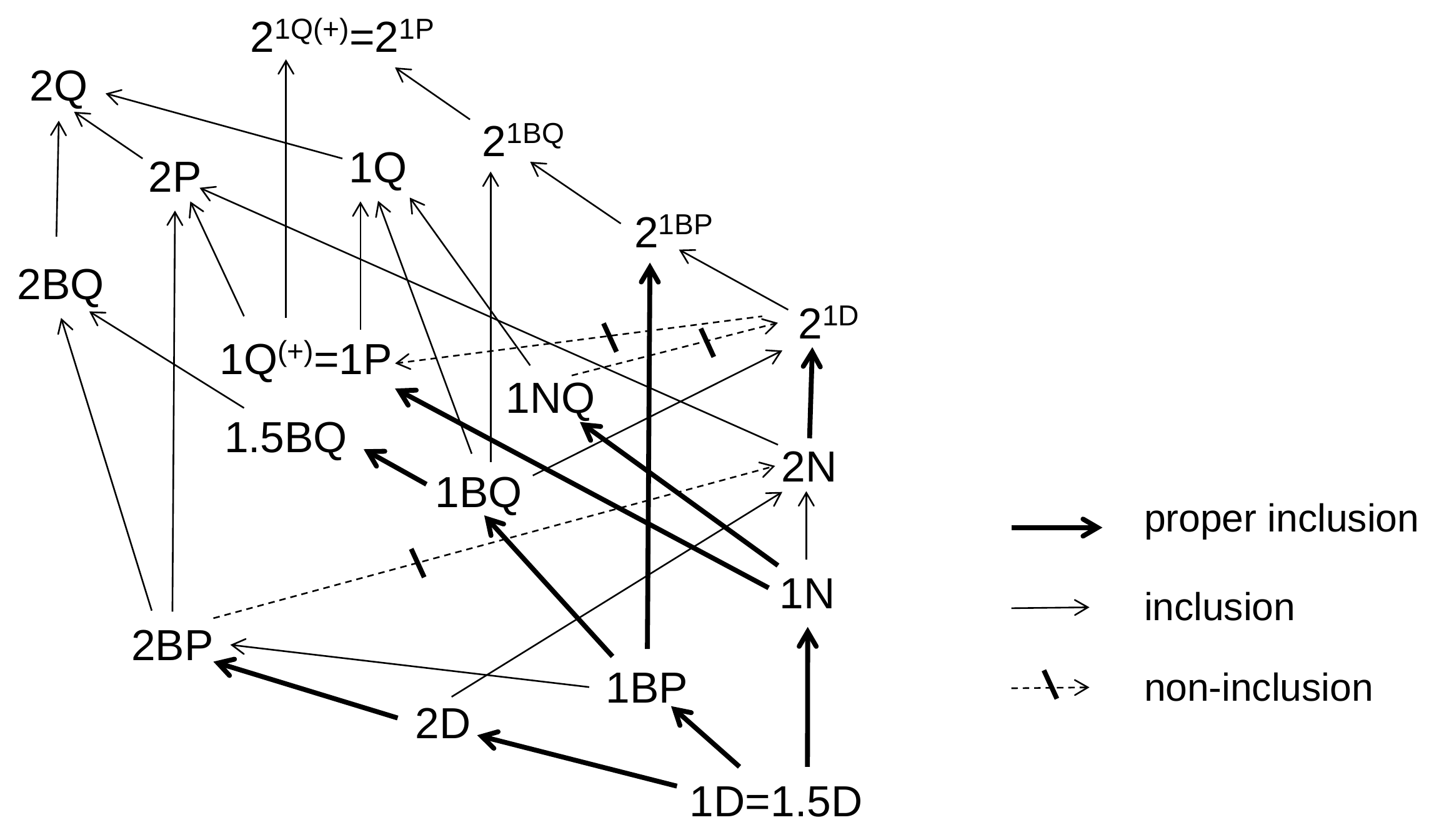}
\caption{Inclusion/separation relationships among nonuniform state complexity classes}\label{fig:relationship}
\end{figure}

Secondly, we are focused on a restricted form of 2qfa's.
Similarly to Boolean circuits, we can limit the length of input strings fed to given finite automata.
In particular, when the input size $|x|$ of each string $x$ in $L_n^{(+)}\cup L_n^{(-)}$ is limited to at most $p(n)$ for a certain absolute  polynomial $p$, we write $\twon/\poly$ and $\twobq/\poly$ respectively instead of $\twon$ and $\twobq$.

For any 2-way (non)deterministic finite automaton (as well as any Turing machine), the \emph{runtime} of such a machine is customarily set to be the length of the shortest accepting path if any, and the length of the shortest rejecting path otherwise. With the use of this runtime definition, 2-way (non)deterministic finite automata all terminate in linear time.
When we handle probabilistic and quantum finite automata, however, it is of significant importance to consider the \emph{average runtime} over all possible computation paths of these machines.
From a concern for execution efficiency, it is reasonable to concentrate on 2qfa's that run in \emph{expected polynomial time} (or \emph{average polynomial time}) rather than 2qfa's with no time bound.
We say that a family $\{M_n\}_{n\in\nat}$ of 2qfa's runs in expected polynomial time if the average runtime of $M_n$ is restricted to a certain fixed polynomial in $(n,|x|)$, where $x$ is any input given to $M_n$.
To emphasize the expected polynomial runtime, we append the prefix ``ptime-'' as in $\pt\twobq$ and $\pt\twobq/\poly$.

To introduce the nonuniformity notion into a model of quantum Turing machine (or QTM, for short), we equip such a QTM with the Karp-Lipton style \emph{advice} as supplemental external information to empower the  QTM (see, e.g., \cite{NY04}). We use the notation $\bql/\poly$ to express the complexity class composed of decision problems that are solvable with bounded-error probability by QTMs using (worst-case) logarithmic space.
For restricted QTMs running in expected polynomial time, we also emphasize this runtime bound with the prefix ``ptime-'' and write $\pt\bpl$ and $\pt\bql$ instead of $\bpl$ and $\bql$.

We wish to present in Theorem \ref{general-theorem} a close connection between advised complexity classes and nonuniform state complexity classes. By our convention of the runtime of (non)deterministic automata and Turing machines, $\pt\twod$, $\pt\twon$, $\pt\dl$, and $\pt\nl$ appearing in the theorem should be understood respectively as $\twod$, $\twon$, $\dl$, and $\nl$.
Concerning quantum and probabilistic computations, for technical reason, we slightly restrict $\twobq$ and $\twobp$ and write $\twobq^{\dagger}$ and $\twobp^{\dagger}$ for these restrictions (see Section \ref{sec:advised-QTM} for the reasoning).

\begin{theorem}\label{general-theorem}
Let $(\AAA,\BB)\in\{(\mathrm{N},\mathrm{D}), (\mathrm{N},\mathrm{BP}^{\dagger}), (\mathrm{N},\mathrm{BQ}^{\dagger}), (\mathrm{BQ}^{\dagger},\mathrm{BP}^{\dagger})\}$.
It then follows that
$\pt2\AAA/\poly \subseteq \pt2\BB$ {iff}
$\pt\AAA\mathrm{L} \subseteq \pt\BB\mathrm{L}/\poly$,
where, when $\AAA=\mathrm{D}$, ``$\mathrm{DL}$'' is understood as ``$\mathrm{L}$''.
\end{theorem}

\begin{corollary}\label{NL-equivalence}
\renewcommand{\labelitemi}{$\circ$}
\begin{enumerate}
  \setlength{\topsep}{-2mm}%
  \setlength{\itemsep}{1mm}%
  \setlength{\parskip}{0cm}%

\item $\twon/\poly\subseteq \pt\twobq^{\dagger}$ iff $\nl\subseteq \pt\bql/\poly$.

\item $\pt\twobq^{\dagger}/\poly \subseteq  \pt\twobp^{\dagger}$ iff $\pt\bql \subseteq \pt\bpl/\poly$.
\end{enumerate}
\end{corollary}

Theorem \ref{general-theorem} follows from the exact characterizations (Proposition \ref{characterize-PBQL/poly}) of parameterized complexity classes in terms of nonuniform state complexity classes, and vice versa. This proposition helps us translate nonuniform state complexity classes, such as $\twod/\poly$, $\twon/\poly$, $\pt\twobp/\poly$, and $\pt\twobq/\poly$, respectively into their corresponding advised parameterized complexity classes, $\para\dl/\poly$, $\para\nl/\poly$, $\para\pt\bpl/\poly$, and $\para\pt\bql/\poly$, where the last class $\para\pt\bql/\poly$, for example, is the collection of parameterized decision problems $(L,m)$ solvable by bounded-error QTMs in expected polynomial time in $|x|m(x)$ using work tapes of space logarithmic in $m(x)$ with (deterministic) advice of size polynomial in $|x|m(x)$ (see Section \ref{sec:parameter-promise} for their precise definitions).

The third main issue concerns quantum advice and an extension of 2qfa's.
The notion of \emph{quantum advice} was introduced by Nishimura and Yamakami \cite{NY04} to enhance the computational capability of polynomial-time QTMs. Quantum advice manifests a quantization of \emph{deterministic advice} and \emph{randomized advice} (see, e.g., \cite{Yam10,Yam14b}).
To emphasize the use of quantum advice, we write $\bql/\qpoly$ in accordance with \cite{NY04}.

In parallel to the change of deterministic advice to quantum advice, we also modify our basic model of 2qfa's as follows. Firstly, we express a quantum  transition function as the form of a matrix or a table, which must be easily encoded into a string over a certain alphabet. For readability, we use the term ``transition table'' to address this encoded string. See Section \ref{sec:transition-table} for the precise definition. This encoding further makes it possible to consider a superposition of transition tables.
Generally, we call by a \emph{super quantum finite automaton} a quantum finite automaton obtained by substituting a superposition of transition tables for the original quantum transition function.

For convenience, we use the notation  $\twosbq$ to express the nonuniform state complexity class obtained from $\twobq$ by substituting super 2qfa's for ordinary 2qfa's and write $\twosbq^{\dagger}$ for its restricted version similarly to $\twobq^{\dagger}$.

\begin{theorem}\label{two-way-equivalence}
$\pt\twosbq^{\dagger}/\poly \subseteq \pt\twobq^{\dagger}$ iff $\pt\bql/\poly = \pt\bql/\qpoly$.
\end{theorem}

A further study on relativizations  (or Turing reducibility) was lately conducted in \cite{Yam19a,Yam19b}.

\section{Preparation: Basic Notions and Notation}\label{sec:preparation}

This section will provide basic notions and notation necessary to read through this exposition.

\subsection{Numbers, Languages, and Hilbert Spaces}

Let $\nat$, $\integer$, and $\complex$ denote respectively the sets of all \emph{natural numbers} (i.e., nonnegative integers), of all \emph{integers}, and of all \emph{complex numbers}. In particular, we set  $\nat^{+}=\nat-\{0\}$, which is the set of all positive integers. Given two integers $m$ and $n$ with $n\geq m$, $[m,n]_{\integer}$ denotes the \emph{integer interval}, which is the set $\{m,m+1,m+2,\ldots,n\}$.
For a given number $n\in\nat^{+}$, we abbreviate $[1,n]_{\integer}$ as $[n]$. For a finite set $S$, $|S|$
indicates the \emph{cardinality} of $S$.
All \emph{polynomials} in this exposition are assumed to have nonnegative integer coefficients, and thus the polynomials are all nondecreasing. Assume also  that all \emph{logarithms} are taken to the base $2$. For our convenience, we set $\log{0}$ to be $0$. A \emph{logarithmic function} is a mapping from $\nat$ to $\nat$ for which there exist two constants $c,d\geq0$ satisfying $f(n)\leq c\log{n}+d$ for any $n\in\nat$.

Let $\Sigma$ be any \emph{alphabet}, which is a finite nonempty set of ``symbols'' or ``letters.''
A \emph{string} over $\Sigma$ is a finite sequence of symbols in $\Sigma$.  The \emph{length} of any string $x$ is denoted $|x|$. In particular, we use the notation $\lambda$ to denote the \emph{empty string} of length $0$. Given a number $n\in\nat$, $\Sigma^n$ (resp., $\Sigma^{\leq n}$) indicates  the set of all strings of length exactly $n$ (resp., of length at most $n$) over $\Sigma$. Let $\Sigma^*= \bigcup_{n\in\nat}\Sigma^n$.
For a string $w$ and a number $i\in[|x|]$, $w_{(i)}$ denotes the $i$th symbol of $w$; for example, $0110_{(1)}=0$ and $0110_{(3)}=1$.
A subset of $\Sigma^*$ is called a \emph{language} over $\Sigma$.
The \emph{complement} of $L$ is the set $\Sigma^*- L$ and is succinctly denoted $\overline{L}$.
We freely identify a language $L$ with its \emph{characteristic function}; that is, $L(x)=1$ for all $x\in L$ and $L(x)=0$ for all $x\notin L$.
Given a size-bounding function $t:\nat\to\nat$, a function $h:\nat\to\Sigma^*$ is called \emph{$t(n)$-bounded} if $|h(n)|\leq t(n)$ holds for all $n\in\nat$.

To describe quantum computation, we use various Hilbert spaces. A \emph{Hilbert space} is a complex inner product space (i.e., a vector space with an inner product), which is also a complete metric space. To express each element of such a space, we use Dirac's notation of $\qubit{\cdot}$. The \emph{dual} of $\qubit{\phi}$ is denoted $\bra{\phi}$. The \emph{norm} of $\qubit{\phi}$ is
denoted $\|\qubit{\phi}\|$, which equals $\sqrt{|\measure{\phi}{\phi}|}$, where $\measure{\phi}{\psi}$ is the \emph{inner product} of $\qubit{\phi}$ and $\qubit{\psi}$.

\subsection{Computational Models of Finite Automata}\label{sec:machine-models}

Our finite automata are always equipped with \emph{read-only input tapes}, which use two endmarkers, $\cent$ (left endmarker) and $\dollar$ (right endmarker), where a given input string is written initially between the two endmarkers.
A \emph{stationary move} refers to a finite automaton's special move by which an input-tape head reads an input symbol but never moves to a neighboring cell.
For clarity reason, we use the term ``one way'' to refer to the condition of a tape head that always moves to the right without stopping, and therefore  there is no stationary move. On the contrary, if we allow a machine to make ``stationary moves,'' we instead use the term ``1.5 way'' to emphasize its  difference from ``one way'' machines.
For brevity, we generally abbreviate ``1-way deterministic finite automaton'' as ``1dfa'', ``1-way nondeterministic finite automaton'' as ``1nfa'', and ``1-way probabilistic finite automaton'' as ``1pfa''. Analogously, for models of ``two way'' tape head moves, we use similar  abbreviations, such as 2dfa, 2nfa, and 2pfa.

We assume the reader's familiarity with the basics of quantum computation (see, e.g., \cite{Gru99,NC00}). Since Kondacs and Watrous's model of 1qfa's \cite{KW97} is strictly weaker in power than 1dfa's, there have been numerous generalizations proposed in the past literature (see, e.g., a survey \cite{AY15}).
We here empower their 1qfa's by simply equipping each of them with a \emph{write-once garbage tape} in which a machine drops any symbol (called a \emph{garbage symbol}) but never accesses any non-blank symbol written already on the garbage tape again.
An early idea of 1qfa's discarding unwanted information down to a portion of a read-once input tape was materialized in \cite[Section 5.2]{Yam14b} and such a mechanism was shown to enhance the computational power of 1qfa's. Yakaryilmaz, Freivalds, Say, and Agadzanyan  \cite{YFSA12} also discussed a similar concept in terms of write-only memory.
The use of a garbage tape makes 1qfa's simulate all 1dfa's simply by discarding  the current inner state of the 1dfa's onto the garbage tape.
The garbage tape can be viewed as a surrounding environment that exists ``externally,'' separated from the essential part of a machine's computation.
Observing the inner state and also the content of the garbage tape at every step produces a mixed state of $M$'s configurations, and therefore, the use of the garbage tape can simulate the behavior of superoperators acting on mixed states. Our model thus turns out to be at least as powerful as the model of \emph{generalized one-way quantum finite automata} given in  \cite{Hir10,YS11}, which allow 1qfa's to use mixed states and superoperators.
Recall that input tapes of finite automata always have the left endmarker $\cent$ and the right endmarker $\dollar$.  All tape cells are indexed by numbers in $\nat$; in particular, $\cent$ is always placed in cell $0$ (which is also called the \emph{start cell}).
For the sake of convenience, we use \emph{circular input tapes}, where a  circular tape is a standard tape whose ends are both glued together so that the right-side of the right endmarker $\dollar$ is the left endmarker $\cent$.

Formally, a \emph{1-way quantum finite automaton with a garbage tape} (where we hereafter use the term ``1qfa'' to indicate this particular model unless otherwise stated) $M$ is an octuple\footnote{Introducing 1qfa's in this form aims at handling both 1qfa's and 2qfa's in the same fashion.}  $(Q,\Sigma,\{\cent,\dollar\},\Xi,\delta,q_0,Q_{acc},Q_{rej})$, where $Q$ is a finite set of inner states, $\Sigma$ is an input alphabet, $\Xi$ is a \emph{garbage alphabet}, $\delta$ is a (quantum) transition function mapping $Q\times\check{\Sigma}\times Q\times \Xi_{\lambda}$ to $\complex$, $q_0$ ($\in Q$) is the initial (inner) state, and $Q_{acc}$ is a set of \emph{accepting states} with $Q_{acc}\subseteq Q$, and $Q_{rej}$ is a set of \emph{rejecting states} with $Q_{rej}\subseteq Q$, where $\check{\Sigma}=\Sigma\cup\{\cent,\dollar\}$ and $\Xi_{\lambda} = \Xi\cup\{\lambda\}$.
For clarity reason, we write $\delta(q,\sigma|p,\xi)$ instead of $\delta(q,\sigma,p,\xi)$. A transition $\delta(q,\sigma|p,\xi)=\alpha$ indicates that, assuming that $M$ is in inner state $q$ scanning input symbol $\sigma$, in a single step with \emph{transition amplitude} $\alpha$, $M$ changes its inner state to $p$, moves the input-tape head to the right, writes down garbage symbol $\xi$ onto the garbage tape (where $\lambda$ means ``no writing''), and moves the garbage-tape head to the right if $\xi\neq\lambda$.
At any step, $M$ freely chooses whether or not it dumps garbage symbols onto its garbage tape and then moves its tape head to the right.
When $M$ does not write any garbage symbol, the garbage-tape head must stay still. For later convenience, this property is referred to as the \emph{flexible use of a garbage tape} or, more succinctly, a \emph{flexible garbage tape}.
Let $Q_{halt} = Q_{acc}\cup Q_{rej}$.
All inner states in $Q_{halt}$ are called \emph{halting states} and the rest of inner states are {\em non-halting states}.

A \emph{configuration} of a 1qfa $M$ on input $x$ is a quadruple $(q,x,i,z)$, where $q\in Q$, $x\in\Sigma^*$, $i\in[0,|x|+1]_{\integer}$, and $z\in\Xi^*$. This describes a ``snapshot'' of the machine's internal status at a certain moment, where $M$ is in inner state $q$, its input-tape head is located at the $i$th tape cell, and the garbage tape contains $z$.
The \emph{initial configuration} of $M$ on $x$ is $(q_0,x,0,\lambda)$ and an \emph{accepting configuration} (resp., a \emph{rejecting configuration}) of $M$ on $x$ is of the form $(q,x,i,z)$ with $q\in Q_{acc}$ (resp., $q\in Q_{rej}$). A \emph{halting configuration} is either an accepting configuration or a rejecting configuration.
Let $\HH_{conf}$ denote the Hilbert space spanned
by all configurations of $M$. Similarly, $\HH_{acc}$ and $\HH_{rej}$ are Hilbert spaces spanned respectively by all accepting configurations and all rejecting configurations of $M$.
For convenience, two more Hilbert spaces, $\HH_{halt}$ and $\HH_{non}$, are defined respectively in terms of halting configurations and non-halting configurations.
The (quantum) transition function $\delta$ naturally induces the \emph{time-evolution operator}  ${U}_{\delta}$ of $M$, which is a linear operator acting on the space $\HH_{conf}$, defined as:
\[
U_{\delta} \qubit{q,x,i,z} = \sum_{(p,\xi)} \delta(q,x_{(i)} | p,\xi) \qubit{p,x,i+1\;(\mathrm{mod}\;|x|+2),z\xi},
\]
where $(p,\xi)$ ranges over $Q\times \Xi_{\lambda}$, each $x_{(i)}$ is the $(i+1)$th symbol of $\cent x\dollar$ with  $x_{(0)}=\cent$ and $x_{(|x|+1)}=\dollar$.

For any subscript $\tau\in\{acc,rej,halt,non\}$, $\Pi_{\tau}$ denotes the \emph{projective measurement} onto the Hilbert space $\HH_{\tau}$. At each step, we first apply ${U}_{\delta}$ and then perform a projective measurement by applying  $\Pi_{acc}\oplus \Pi_{rej}$. A \emph{computation} of $M$ on input $x$ is a series $(\qubit{\phi_0},\qubit{\phi_1},\ldots,\qubit{\phi_{|x|+2})}$ of superpositions of configurations of $M$ on $x$ defined as  $\qubit{\phi_0}=\qubit{q_0,x,0,\lambda}$ and $\qubit{\phi_{i+1}} = \Pi_{non}U_{\delta}\qubit{\phi_i}$ for each index $i\in[0,|x|+1]_{\integer}$.
At Step $i\in[1,|x|+2]_{\integer}$, we say that $M$ \emph{accepts} (resp., \emph{rejects}) $x$ with probability $p_{acc}(x,i) = \|\Pi_{acc}U_{\delta}\qubit{\phi_{i-1}}\|^2$ (resp., $p_{rej}(x,i)= \|\Pi_{rej}U_{\delta}\qubit{\phi_{i-1}}\|^2$). The \emph{(cumulative) acceptance probability} $p_{acc}(x)$ of $M$ on $x$ is the sum $\sum_{i=1}^{|x|+2}p_{acc}(x,i)$ and the \emph{(cumulative) rejection probability} of $M$ on $x$ is $\sum_{i=1}^{|x|+2}p_{rej}(x,i)$.

When $x$ is fixed through our analysis of quantum computation, we often remove $x$ from the configurations and instead consider  \emph{surface configurations} $(q,i,z)$. In place of $U_{\delta}$, we use the notation  $U^{(x)}_{\delta}$, which acts on the Hilbert space spanned by all surface configurations.

As a natural expansion of 1qfa's, we define a \emph{2-way quantum finite automaton with a garbage tape} (or a 2qfa, for short) by allowing a tape head to move in both directions as well as to stay still. To be more formal, a 2qfa $M$ is of the form $(Q,\Sigma,\{\cent,\dollar\},\Xi,\delta,q_0,Q_{acc},Q_{rej})$ with a (quantum) transition function $\delta: Q\times \check{\Sigma}\times Q\times D\times \Xi_{\lambda}\to\complex$ with $D=\{-1,0,+1\}$, where ``$-1$'' and ``$+1$'' indicate that an input-tape head moves respectively to the left and to  the right, and ``$0$'' means that the tape head makes a stationary move.   Assuming that $M$ is in inner state $q$, scanning input symbol $\sigma$, if we make a transition  $\delta(q,\sigma|p,d,\xi)=\alpha$, then $M$ changes its inner state to $p$, moves its input-tape head in direction $d$, and writes down garbage symbol $\xi$ by moving its tape head to the right whenever  $\xi\neq\lambda$.
Similarly to 1qfa's, the use of garbage tapes provides sufficiently high computational power to 2qfa's.

The notion of configurations for 2qfa's is the same as that for 1qfa's; however, the time-evolution operator $U_{\delta}$ is in the form:
\[
U_{\delta} \qubit{q,x,i,z} = \sum_{(p,d,\xi)} \delta(q,x_{(i)} | p,d,\xi) \qubit{p,x,i+d\;(\mathrm{mod}\;|x|+2),z\xi},
\]
where $(p,d,\xi)$ ranges over $Q\times D\times \Xi_{\lambda}$.
When $x$ is fixed, we write $U^{(x)}_{\delta}$ for the associated time-evolution operator on the Hilbert space spanned by surface configurations.
Unlike the case of 1qfa's, no time bound is in general imposed on each 2qfa, and thus some of the computation paths may possibly be infinitely long.
To scale the runtime of such a 2qfa, we thus need to take the expectation of the lengths of computation paths; that is, $\sum_{i=1}^{\infty} i \cdot \|\Pi_{halt} U_{\delta}\qubit{\phi_{i-1}}\|^2$, where $\qubit{\phi_i}$ is the superposition of configurations obtained after the $i$th step of the 2qfa. We call this value the \emph{expected runtime} of $M$ on $x$.

We say that $M$ is \emph{well-formed} if ${U}_{\delta}^{(x)}$ is a unitary matrix (i.e., $U^{(x)}_{\delta}(U^{(x)}_{\delta})^{\dagger}=I$) for all $x\in\Sigma^*$.
In the rest of this exposition, we always assume that all 1qfa's as well as 2qfa's are well-formed.

\subsection{Transition Tables}\label{sec:transition-table}

The behavior of a finite automaton is completely dictated by its  transition function $\delta$. However, it is sometimes convenient to use ``transition tables'' instead of transition functions.
A \emph{transition table} is intuitively a matrix form of the ``description'' of $\delta$, which is further encoded into a classical string.
Each row of a transition table is indexed by the elements $(q,\sigma)$ in $Q\times\check{\Sigma}$, each column is indexed by the elements $(p,d)$ in $Q\times D$, and the $((q,\sigma),(p,d))$-entry of the table contains $1$ if $\delta$ maps $(q,\sigma)$ to $(p,d)$; $0$ otherwise.
Although this definition is valid for deterministic/nondeterministic finite automata, we cannot use the same one for 2qfa's because we need to deal with a set of (quantum) transition amplitudes, which are generally arbitrary complex numbers. Hence, we need to find an appropriate way of fitting such transition amplitudes into transition tables.

Let $M$ denote a 2qfa of the form $(Q,\Sigma,\{\cent,\dollar\},\Xi, \delta,q_0,Q_{acc},Q_{rej})$ with $\Sigma$ and $\Xi$ of constant sizes.
Since $M$'s transition amplitudes are complex numbers, we want to use a \emph{quantum circuit} to generate (or approximate) those amplitudes and we then encode this quantum circuit into a transition table, where a quantum circuit is made up of finitely many quantum gates taken from a certain \emph{universal set}. In the rest of this exposition, we fix $\{CNOT,H,T_{\pi/8}\}$ as such a universal set, where $CNOT$ is the Controlled-NOT, $H$ is the Hadamard transform, and $T_{\pi/8}$ is the $\pi/8$-rotation around the $z$ axe (see, e.g., \cite{NC00} for their definitions and properties).
If necessary, we also use the \emph{identity map} $I$ as a special quantum gate.

Formally, we express inner states of $M$ as strings in $\{0,1\}^{r_1}$, symbols in $\Xi_{\lambda}$ as strings in $\{0,1\}^{r_2}$ for two  appropriate numbers $r_1$ and $r_2$ in $\nat^{+}$ (thus, $|Q|=2^{r_1}$ and $|\Xi_{\lambda}|=2^{r_2}$), and head directions $\{-1,0,+1\}$ as elements  $\{10,11,01\}$, respectively. For convenience, we write $V_{q,\sigma}$ for a $2^{r_1+r_2+2}\times 2^{r_1+r_2+2}$ unitary matrix that, on input $\qubit{\phi_0} = \qubit{0^{r_1}}\qubit{00}\qubit{0^{r_2}}$, produces a quantum state $\sum_{(p,d,\xi)} \delta(q,\sigma|p,d,\xi) \qubit{p,d,\xi}$.
A \emph{transition table} $T$ of $M$ on input $x$ is a matrix, each row of which is indexed by $(q,\sigma)$ in $Q\times \check{\Sigma}$, whose $(q,\sigma)$-row contains a ``description'' of
$V_{q,\sigma}$.
Given any parameter $n$, any input length $l$, and any pair $(q,\sigma)\in Q\times \check{\Sigma}$, we intend to define a quantum circuit $C^{(n,l)}_{q,\sigma}$ so as to approximate $V_{q,\sigma}$.

Assume that $M$ on input $x$ runs in expected $p(n,|x|)$ time for a certain function $p$ and errors with probability at most a certain constant
$\varepsilon\in [0,1/2)$.
It suffices to consider the first $cp(n,|x|)$ steps of $M$ for an appropriately chosen absolute constant $c\geq1$ to guarantee that the error probability obtained during these steps is still at most another constant $\varepsilon' = \frac{1}{2}(\varepsilon+\frac{1}{2})$, where $\varepsilon<\varepsilon'<\frac{1}{2}$.
Under this circumstance, letting $\alpha =\frac{1}{2}(\frac{1}{2}-\varepsilon)$, we want to approximate $V_{q,\sigma}$ with \emph{inaccuracy} of $\alpha 2^{-(r_1+r_2+2)}/cp(n,|x|)$ by a certain quantum circuit $C^{(n,l)}_{q,\sigma}$ made up of the aforementioned universal quantum gates; namely, $\|V_{q,\sigma}\qubit{\phi_0} - C^{(n,|x|)}_{q,\sigma}\qubit{\phi_0}\|\leq \alpha 2^{-(r_1+r_2+2)}/cp(n,|x|)$.  For convenience, we write $\tilde{M}$ for the machine obtained from $M$ by replacing each $V_{q,\sigma}$ with $C^{(n,l)}_{q,\sigma}$.

Let $s=s(n,l)$ denote the total number of the quantum gates used in $C^{(n,l)}_{q,\sigma}$. A reasonable upper bound of $s(n,l)$ is given by the following lemma.

\begin{lemma}\label{Solovay-Kitaev}
Any $2^r\times2^r$ unitary matrix $V$ can be approximated by a certain $r$-qubit quantum circuit $C$ of $O(n^2r^{4}2^{2r})$ universal gates satisfying $\|V\qubit{0^r} - C\qubit{0^r}\|\leq 2^{-n}$.
\end{lemma}

\begin{proof}
Given a unitary matrix $V$, as in the same way described in \cite[Section 4.5.1]{NC00}, we can take a number $k\leq 2^{r-1}(2^r-1)$ and $k$
2-level unitary matrices $U_1,U_2,\cdots,U_{k}$ yielding $V= U_1U_2\cdots U_{k}$. In the same way as in \cite[Section 4.5.2]{NC00},
we then decompose
each 2-level unitary matrix $U_i$ into $O(r^2)$ 1-qubit and CNOT gates, where 1-qubit gates may not be limited to $\{H,T_{\pi/8}\}$.  Combining them, we can realize $V$ by a certain quantum circuit of $O(r^22^{2r})$ 1-qubit and CNOT gates. Let $s$ denote the number of used quantum gates.
We set $\varepsilon=2^{-n}/s$.
The Solovay-Kitaev theorem (e.g., \cite{KSV02,NC00}) then guarantees the existence of a constant $e$ in the real interval $[1,2]$ such that each 1-qubit gate can be approximated by $O(\log^e(1/\varepsilon))$ universal gates from $\{CNOT,H,T_{\pi/8}\}$ to within $\varepsilon$.
To approximate $V$, we need $O(r^22^{2r})\times O(\log^e(r^22^{n+2r})) \subseteq O(n^2r^{4}2^{2r})$ universal gates because of $1\leq e\leq 2$, $\log(r^2)\leq 2r$, and $(n+2r)^e\leq n^e(2r)^e \leq 4n^2r^2$.
\end{proof}

Since $V_{q,\sigma}$ is a $2^{r_1+r_2+2}\times 2^{r_1+r_2+2}$ unitary matrix,
Lemma \ref{Solovay-Kitaev} implies that $s(n,l)$ is bounded by  $O((r_1+r_2+2)^6 2^{2(r_1+r_2+2)} \log^2p(n,l)) \subseteq O(|Q|^2|\Xi|^2 \log^6|Q||\Xi|\cdot \log^2p(n,l)) \subseteq O(|Q|^8|\Xi|^8\log^2p(n,l))$ because of $\log|Q||\Xi|\leq |Q||\Xi|$.
Thus, we can express the quantum circuit $C^{(n,l)}_{q,\sigma}$ as an appropriate series $(I^{k_1}\otimes G_1\otimes I^{k'_1}) (I^{k_2}\otimes G_2\otimes I^{k'_2}) \cdots (I^{k_s}\otimes G_s\otimes I^{k'_s})$, where $s=s(n,l)$, $k_i\in[|Q|+2]$, and $G_i \in  \{I,CNOT,H,T_{\pi/8}\}$ for any index $i\in[s]$.
From this series, $C^{(n,l)}_{q,\sigma}$ is completely specified by the series  $(k_1,G_1),(k_2,G_2),\ldots,(k_s,G_s)$, and thus $C^{(n,l)}_{q,\sigma}$ can be encoded into a string $\pair{C^{(n,l)}_{q,\sigma}}$ of the form $1^{k_1}\#\pair{G_1} \#^21^{k_2} \#\pair{G_2} \#^2\cdots \#^21^{k_s}\#\pair{G_s}$, where $\pair{I} = 1$, $\pair{CNOT}=2$, $\pair{H}=3$, and $\pair{T_{\pi/2}}=4$.
Note that, if this string $\pair{C^{(n,l)}_{q,\sigma}}$ is given to a tape, then it is possible to generate the quantum state $C^{(n,l)}_{q,\sigma}\qubit{0^{r_1}}\qubit{00}\qubit{0^{r_2}}$ by ``sweeping''  the tape once from left to right and applying operators $I^{k_i}\otimes G_i \otimes I^{k'_i}$ one by one.

We then expand our encoding to all quantum transitions. Let $\{C^{(n,l)}_{a,b}\}_{(a,b)\in Q\times\check{\Sigma}}$ be a set of all approximated quantum circuits corresponding to $\delta$.
We first enumerate all elements in $Q\times\check{\Sigma}$ as $\{(a_1,b_1),(a_2,b_2),\ldots,(a_k,b_k)\}$, where $k=|Q\times\check{\Sigma}|$, and, according to this enumeration, we set  $\pair{T}$ as $\pair{C^{(n,l)}_{a_1,b_1}}\#^3\pair{C^{(n,l)}_{a_2,b_2}}\#^3\cdots \#^3 \pair{C^{(n,l)}_{a_k,b_k}}$. The \emph{length} of the string $\pair{T}$ is bounded from above by $k\cdot O(s) \subseteq O(|Q|^9|\Xi|^8 \log^2p(n,l))$.

Since $M$ preforms at most $cp(n,|x|)2^{r_1+r_2+2}$ applications of the matrices  $V_{q,\sigma}$, the quantum state produced by $M$ after $cp(n,|x|)$ steps can be approximated to within $\alpha$ by the quantum state produced by $\tilde{M}$ in $cp(n,|x|)$ steps. This implies that the difference between the acceptance (resp., rejecting) probabilities between $M$ and $\tilde{M}$ is at most   $\alpha$ (see, e.g., \cite{BV97,Yam03}). Therefore, if $M$ accepts (resp., rejects) $x$ with probability at least $1-\varepsilon$, then $\tilde{M}$ accepts (resp., rejects) $x$ with probability at least $1-\varepsilon-\alpha$, which equals $1-\varepsilon'$.

\subsection{Quantum Turing Machines with Advice}

Each Turing machine is equipped with a read-only input tape with the two endmarkers, $\cent$ and $\dollar$, as well as a \emph{rewritable work tape}. Occasionally, we further equip a Turing machine with a \emph{read-only advice tape}, which holds a given advice string surrounded by the same endmarkers. It is important to note that no machine modifies an advice string during its computation.

In accordance with the aforementioned garbage-tape quantum finite automata, we also supply quantum Turing machines with \emph{(flexible) garbage tapes}. Since we discuss only such machines in later sections, \emph{quantum Turing machines equipped with (flexible) garbage tapes} are simply referred to as QTMs. later in Section \ref{sec:one-way-QFA}, we will introduce
another notion of ``rigid'' garbage tape.
A QTM also has a work tape and a work alphabet $\Gamma$ (including a unique \emph{blank symbol} $B$) and both the work and the garbage tapes are initially blank.
We further provide a piece of useful  information, known as  ``advice.''
An \emph{advice function} $h$ is a function from $\nat$ to $\Theta^*$ for a certain advice alphabet $\Theta$, and each value $h(n)$ is called an \emph{advice string}.
Since we need to handle such advice, we further furnish the QTM with a distinguished \emph{advice tape}. Each advice string is initially written on the advice tape, surrounded by the two endmarkers, $\cent$ and $\dollar$.
For convenience, we call a QTM with an advice
tape by an \emph{advised QTM}.
Formally, an advised  QTM is a decuple $(Q,\Sigma,\{\cent,\dollar\},\Gamma, \Theta, \Xi, \delta,q_0,Q_{acc},Q_{rej})$, including a work alphabet $\Gamma$, an advice alphabet $\Theta$, a garbage alphabet $\Xi$, and $\delta$ is a (quantum) transition function from $Q\times\check{\Sigma}\times\Gamma \times \check{\Theta} \times Q\times \Gamma\times \Xi_{\lambda} \times D_1\times D_2\times D_3$ to $\complex$, where $\check{\Theta} = \Theta\cup\{\cent,\dollar\}$, $\Xi_{\lambda} = \Xi\cup\{\lambda\}$, $D_1=D_2=\{-1,0,+1\}$ (for the input-tape and the work-tape heads), and $D_3=\{+1\}$ (for the garbage-tape head). Notice that our QTM can take arbitrary complex amplitudes, not limited to polynomial-time approximable amplitudes as in most literature.

A \emph{configuration} of $M$ is a tuple $(q,x,t_1,y,t_2,w,t_3,z)$, where $q\in Q$, $t_1,t_2,t_3\in\integer$, $x\in \Sigma^*$, $y\in\Gamma^*$, $w\in\Theta^*$, and $z\in \Xi^*$. This configuration expresses a circumstance in which $M$ is in inner state $q$, scanning the $t_1$th cell of the input tape, the $t_2$th cell of the work tape containing $y$, the $t_3$th cell of the advice tape containing $w$, and the garbage tape containing $z$. There is no need to include the head position of the garbage tape similarly to the 2qfa case.
When $x$ and $h(|x|)$ are fixed throughout computation, we use a \emph{surface configuration} of the form $(q,t_1,y,t_2,t_3)$ instead.
In a way similar to 2qfa's, we define  $\HH_{acc}$, $\HH_{rej}$, $\HH_{halt}$, and $\HH_{non}$ as well as $\Pi_{acc}$, $\Pi_{rej}$, $\Pi_{halt}$, and $\Pi_{non}$.
The \emph{configuration space} of $M$ is denoted $\HH_{conf}$.
The \emph{time-evolution operator} of an advised QTM acting on  $\HH_{conf}$ is defined in a similar way as 2qfa's;
that is,
\begin{eqnarray*}
\lefteqn{U_{\delta} \qubit{q,x,t_1,y,t_2,w,t_3,z}} \hs{5} \\
&=& \sum_{(p,\tau,d_1,d_2,d_3,\xi)} \zeta_{\delta} \qubit{p,x,t_1+d_1 \:\mathrm{mod}\:(|x|+2),\tilde{y}_{\xi}, t_2+d_2,w,t_3+d_3, z\xi},
\end{eqnarray*}
where $\zeta_{\delta} = \delta(q,x_{(t_1)},y_{(t_2)},w_{(t_3)},  p,\tau,d_1,d_2,d_3,\xi)$, $\tilde{y}_{\xi}= y_{(1)}y_{(2)}\cdots y_{(t_2-1)} \xi y_{(t_2+1)}\cdots y_{(|y|)}$, and $(p,\tau,d_1,d_2,d_3,\xi)$ ranges over $Q\times \check{\Theta} \times D_1\times D_2\times D_3\times \Xi_{\lambda}$.
We also use $U^{(x,h(|x|))}_{\delta}$ instead of $U_{\delta}$
when $x$ and $h(|x|)$ are fixed throughout computation.
We demand that the time-evolution operator $U^{(x,h(|x|))}_{\delta}$ of our QTM should be \emph{unitary} for any input $x\in\Sigma^*$.
A \emph{computation} of $M$ on input $x$ is defined, similarly to 2qfa's,  as a series of superpositions of configurations generated by $U_{\delta}$ and $\Pi_{non}$.
The \emph{expected runtime} of $M$ with $h$ on $x$, denoted by $\mathrm{Time}_{M,h}(x)$, is the expectation of the lengths of the computation paths of $M$ on $x$.
We say that $M$ with $h$ runs in \emph{expected polynomial time} if there exists a polynomial $p$ satisfying $\mathrm{Time}_{M,h}(x)\leq p(|x|)$ for all $x$. The \emph{space usage} of $M$ on $x$, denoted by $\mathrm{Space}_{M,h}(x)$, is the maximal cell index that $M$'s work-tape head visits during $M$'s computation on $x$ with $h$. We say that $M$ with $h$ uses \emph{logarithmic space} (or log space) if there exist two constants $c,d\geq 0$ satisfying $\mathrm{Space}_{M,h}(x)\leq c\log{|x|}+d$ for all $x$.

Given a language $L$ over $\Sigma$, we say in general that a machine $M$ \emph{recognizes $L$ with bounded-error probability} if there exists a constant $\varepsilon\in[0,1/2)$ such that, for any $x\in L$, $M$ accepts $x$ with probability at least $1-\varepsilon$ and, for any $x\in\Sigma^*-L$, $M$ rejects $x$ with probability at least $1-\varepsilon$.  In contrast, $M$ is said to \emph{recognize $L$ with unbounded-error probability} if, for any $x\in L$, $M$ accepts $x$ with probability more than $1/2$ and, for any $x\in\Sigma^*-L$, $M$ rejects $x$ with probability at least $1/2$.

The advised quantum complexity class $\bql/\poly$ consists of all languages, each of which is recognized with bounded-error probability by a certain QTM equipped with an advice tape and a polynomially-bounded advice function using only logarithmic space. In a similar manner, we define $\bpl/\poly$ and $\nl/\poly$ respectively
using probabilistic Turing machines (PTMs) and nondeterministic Turing machines (NTMs).
When underlying QTMs and PTMs are further restricted to run in expected polynomial time, we obtain $\pt\bql/\poly$ and $\pt\bpl/\poly$ from $\bql/\poly$ and $\bpl/\poly$, respectively.

In Section \ref{sec:parameter-promise}, we will further transform those complexity classes to their ``parameterized'' versions.

\subsection{Promise Decision Problems and Parameterized Decision Problems}\label{sec:parameter-promise}

A \emph{promise decision problem} in general has the form $(A,B)$ over a given input  alphabet $\Sigma$ satisfying both $A,B\subseteq\Sigma^*$ and $A\cap B=\setempty$.
In the past literature \cite{AY12,BMP14b,ZLQG17}, promise problems have been studied also in connection to quantum finite automata.
As stated in Section \ref{sec:nonuniform-sc}, we particularly deal with a ``family'' $\{(L_n^{(+)},L_n^{(-)})\}_{n\in\nat}$ of promise decision problems over a certain fixed alphabet $\Sigma$ (not depending on $n$), where $L_n^{(+)}$ consists of ``positive'' instances and $L_n^{(-)}$ consists of ``negative'' instances.
Notice that, for each index $n\in\nat$, $L_n^{(+)}\cap L_n^{(-)}=\setempty$ and $L_n^{(+)}\cup L_n^{(-)}\subseteq \Sigma^*$.
However, we usually do not demand $(L_n^{(+)}\cup L_n^{(-)})\cap (L_m^{(+)}\cup L_m^{(-)})=\setempty$ for any distinct pair $(m,n)$.
For any family $\LL = \{(L_n^{(+)},L_n^{(-)})\}_{n\in\nat}$ of promise decision problems, we use the notation $\co\LL$ to denote the family  $\{(L_n^{(-)},L_n^{(+)})\}_{n\in\nat}$ obtained by exchanging between $L_n^{(+)}$ and $L_n^{(-)}$.

Let $\LL=\{(L_n^{(+)},L_n^{(-)})\}_{n\in\nat}$ denote any family of promise decision problems. Given a family $\MM=\{M_n\}_{n\in\nat}$ of certain ``specified'' machines that satisfy certain ``predetermined'' criteria for acceptance and rejection, we generally say that $M_n$ \emph{recognizes} (\emph{solves} or \emph{computes}) the promise decision problem $(L_n^{(+)},L_n^{(-)})$ if (1) for any positive instance $x\in L_n^{(+)}$, $M_n$ accepts $x$ and, (2) for any negative instance $x\in L_n^{(-)}$, $M_n$ rejects $x$. There is no requirement for the behavior of $M_n$ on any string $x$ outside of $L_n^{(+)}\cup L_n^{(-)}$ and $M_n$ may possibly \emph{neither accept nor reject} such a string $x$.
Conveniently, we set $\Sigma_n = L_n^{(+)}\cup L_n^{(-)}$ for each index $n\in\nat$. Any input in $\Sigma_n$ is said to be a \emph{valid input}. We also say that $x$ is \emph{promised} if $x$ is a valid input.
Note that we \emph{do not} force $\{\Sigma_n\}_{n\in\nat}$ to satisfy $\Sigma_{n}\cap \Sigma_{n'} =\setempty$ for all distinct pairs $n,n'\in\nat$.
The  family $\MM$ of machines is said to \emph{solve} $\LL$ with \emph{bounded-error probability} if there exists a constant $\varepsilon\in[0,1/2)$ such that, for any index $n\in\nat$, (i) for all $x\in L_n^{(+)}$, $p_{acc,M_n}(x)\geq 1-\varepsilon$ and (ii) for all $x\in L_n^{(-)}$, $p_{rej,M_n}(x)\geq 1-\varepsilon$. In comparison,  $\MM$ \emph{solves $\LL$ with unbounded-error probability} if, for any $n\in\nat$, (i$'$) for all $x\in L_n^{(+)}$, $p_{acc,M_n}(x)>\frac{1}{2}$ and (ii$'$) for all $x\in L_n^{(-)}$, $p_{rej,M_n}(x)\geq \frac{1}{2}$.

For two families $\LL=\{(L_n^{(+)},L_n^{(-)})\}_{n\in\nat}$ and $\hat{\LL}=\{(\hat{L}_n^{(+)},\hat{L}_n^{(-)})\}_{n\in\nat}$ of promise decision problems, we say that $\hat{\LL}$ is an \emph{extension} of $\LL$ if  $L_n^{(+)}\subseteq \hat{L}_n^{(+)}$ and $L_n^{(-)}\subseteq \hat{L}_n^{(-)}$ hold for any index $n\in\nat$. If the set $\{1^n\# x\mid n\in\nat, x\in L_n^{(+)}\cup L_n^{(-)}\}$ belongs to $\dl$, then $\LL$ is called \emph{$\dl$-good}.
Moreover, a collection $\FF$ of families of promise decision problems is \emph{$\dl$-good} if every element in $\FF$ has an $\dl$-good extension of it in $\FF$.

A \emph{size parameter} is a function from $\Sigma^*$ to $\nat$ for a certain alphabet $\Sigma$. As typical examples, $m_{bin}(x)$ denotes the binary size $|x|$ of input $x$ and $m_{ver}(G)$ indicates the number of vertices in a given graph $G$. A \emph{parameterized decision problem} over $\Sigma$ is a pair $(L,m)$ with a language (equivalently, a decision problem) $L$ over $\Sigma$ and a size parameter $m:\Sigma^*\to\nat$.

Let us define a useful translation between a parameterized decision problem and a family of promise decision problems.
Given a parameterized decision problem $(L,m)$ over an alphabet $\Sigma$, a family $\LL=\{(L_n^{(+)},L_n^{(-)})\}_{n\in\nat}$ of promise decision problems is said to be \emph{induced from} $(L,m)$ if, for each index $n\in\nat$, $L_n^{(+)} = L\cap \Sigma_n$ and $L_n^{(-)} = \overline{L}\cap \Sigma_n$, where $\Sigma_n=\{x\in\Sigma^*\mid m(x)=n\}$.

On the contrary, let $\LL = \{(L_n^{(+)},L_n^{(-)})\}_{n\in\nat}$ be an $\dl$-good family of promise decision problems over a fixed alphabet $\Sigma$. We set $L_{all} = \bigcup_{n\in\nat}(L_n^{(+)}\cup L_n^{(-)})$, which is obviously included in  $\Sigma^*$ but may not equal $\Sigma^*$. With the use of a distinguished separator $\#$, we write $\Sigma_{\#}$ for the set $\Sigma\cup\{\#\}$. For each index $n\in\nat$, we define $K_{n}^{(+)} = \{1^n\# x\mid x\in L_n^{(+)}\}$ and $K_n^{(-)} = \{1^n\# x\mid x\in L_n^{(-)}\}\cup \{z\# x\mid z\in\Sigma^n-\{1^n\},x\in \Sigma_{\#}^*\} \cup \{z\mid z\in \Sigma^n\}$. Furthermore, we set $K=\bigcup_{n\in\nat}K_n^{(+)}$ and $\overline{K}=\bigcup_{n\in\nat} K_n^{(-)}$. It follows that $K\cap \overline{K}=\setempty$ and $K\cup \overline{K}=\Sigma_{\#}^*$. We define a new  size parameter $m':\Sigma_{\#}^*\to\nat$ by setting,  for any $w\in\Sigma_{\#}^*$,  $m'(w)=n$ if $w=1^n\# x$ for a certain $x\in L_n^{(+)}\cup L_n^{(-)}$, and $m'(w)=|w|$ otherwise. Since $\LL$ is $\dl$-good, $m'$ must be computed using logarithmic space.
The pair $(K,m')$ therefore turns out to be a parameterized decision problem over $\Sigma_{\#}$ and  $(K,m')$ is said to be \emph{induced from} $\LL$.

A size parameter $m:\Sigma^*\to\nat$ is said to be \emph{polynomially bounded} if there exists a polynomial $p$ witnessing $m(x)\leq p(|x|)$ for all $x\in\Sigma^*$; in contrast, $m$ is \emph{polynomially honest} if a certain fixed polynomial $q$ satisfies $|x|\leq q(m(x))$ for any $x\in\Sigma^*$. For the size parameter $m'$ defined above, since $m'(w)\leq |w|$ holds for all $w$, $m'$ must be polynomially bounded.
We use the notation $\phsp$ to denote the set of all parameterized decision problems $(L,m)$ for which $m$ is polynomially honest.

We say that $m$ is a \emph{log-space size parameter} if there exists a deterministic Turing machine (or a DTM, for short) $M$ such that, for any string $x$, $M$ takes $x$ as an input and produces $1^{m(x)}$ (i.e., $m(x)$ in unary) on its write-once output tape using $O(\log{|x|})$ work space \cite{Yam17a}. Notice that the function $f(x)=1^{m(x)}$ is polynomially bounded because, otherwise, a log-space machine computing $f$ must stay in an infinite loop. Hence, $m$ is polynomially bounded.

As noted in Section \ref{sec:nonuniform-sc}, the prefix ``para-'' is used to describe parameterized complexity classes. We define $\para\bql$ as the class of parameterized decision problems $(L,m)$ solvable by bounded-error QTMs using $O(\log{m(x)})$ space, where $m$ is any log-space size parameter.
If the expected runtime of each underlying QTM is further limited to $p(m(x),|x|)$ for a fixed polynomial $p$, we use the notation of  $\para\pt\bql$. The probabilistic counterparts of $\para\bql$ and $\para\pt\bql$ are respectively denoted by $\para\bpl$ and $\para\pt\bpl$. With the use of deterministic and nondeterministic Turing machines instead,
their runtime is customarily set to be the length of the shortest accepting path if it exists, and the length of the shortest rejecting path otherwise. Following \cite{Yam17a},
we similarly define  $\para\dl$ and $\para\nl$ as the parameterizations of $\dl$ and $\nl$, respectively.
Moreover, we write   $\para\nl/\poly$ to denote the parameterization of $\nl/\poly$, which is obtained by replacing languages $L$ in $\nl$ with their associated parameterized decision problems $(L,m)$ and also by taking advice functions of advice size at most polynomials in $m(x)|x|$.
Similarly, we obtain $\para\bql/\poly$, $\para\pt\bql/\poly$, etc.
Refer to \cite{Yam17a,Yam18,Yam19a} for the relevant notions.

\subsection{Nonuniform State Complexity}\label{sec:complexity-class}

Our purpose is to introduce nonuniform complexity classes defined by state complexities of quantum finite automata families. Related to these classes, we also consider classes based on deterministic, nondeterministic, and probabilistic finite automata.

The \emph{state complexity} generally refers to the number of inner states used to describe the behavior of a given automaton. However, since we use a (uniform or nonuniform) family $\MM=\{M_n\}_{n\in\nat}$ of finite automata, the state complexity of such a family becomes a function in $n$. More formally, the \emph{state complexity} $sc(n)$ (or $sc(M_n)$) of a family $\MM$ of finite automata $M_n$ with a set $Q_n$ of inner states is a function defined by $sc(n)=|Q_n|$ for all indices $n\in\nat$ \cite{VY15}. In later sections, nonetheless, we use nonuniform families $\{M_n\}_{n\in\nat}$ of finite automata, in which each $M_n$ may possibly be chosen independently, and therefore we emphatically call $sc(n)$ the \emph{nonuniform state complexity function}. Customarily, ``uniformity'' is thought as a special case of ``nonuniformity.''

The nonuniform state complexity class $\oned$ is the collection of all nonuniform families $\{(L_n^{(+)},L_n^{(-)})\}_{n\in\nat}$ of promise decision problems, each family of which is based on a certain fixed alphabet $\Sigma$ (not depending on $n$) and satisfies the following condition: there exist a polynomial $p$ and a nonuniform family $\{M_n\}_{n\in\nat}$ of 1dfa's such that, for each index $n\in\nat$, (i) $M_n$ has at most $p(n)$ inner states and (ii) $M_n$ solves $(L_n^{(+)},L_n^{(-)})$
on all inputs. In a similar way, we can define $\onen$ using 1nfa's instead of 1dfa's.


The use of $n^{O(1)}$ inner states endows underlying 1dfa's with enormous computational power. As a quick example, let us consider the family $\{(L_n^{(+)},L_n^{(-)})\}_{n\in\nat}$ with $L_n^{(+)} =\{a^ib^{i^2}\in \Sigma_n \mid i\leq n\}$, and $L_n^{(-)} = \Sigma_n - L_n^{(+)}$ for any index $n\in\nat$, where $\Sigma_n=\{a^ib^j\mid i,j\in\nat, i,j<2^n\}$.
No standard pushdown automaton recognizes the language $L'=\bigcup_{n\in\nat} L_n^{(+)}$; however, the following 1dfa $M_n$ with $O(n)$ inner states easily solves $(L_n^{(+)},L_n^{(-)})$: on input $x\in\Sigma_n$, find a number $i\leq n$ for which $x=a^ib^j$ and, using this value $i$ repeatedly, check whether $j=i^2$.

The notation $2^{\oned}$ indicates the collection of all nonuniform families of promise decision problems satisfying the following: for each  $\{(L_n^{(+)},L_n^{(-)})\}_{n\in\nat}$ of such families, every promise decision problem $(L_n^{(+)},L_n^{(-)})$ is recognized by a certain 1dfa of at most $2^{p(n)}$ inner states for a certain fixed polynomial $p$ independent of $n$.
If nonuniform families of 2dfa's that have polynomially many inner states are used instead, we obtain $\twod$ from $\oned$. In a similar manner, the use of nondeterministic finite automata introduces $\onen$, $2^{\onen}$, and $\twon$ as well.

In what follows, we present a useful lemma, which directly follows from  \cite[Lemma 3.3]{Yam18}. This lemma will be used in later sections.

\begin{lemma}{\rm \cite{Yam18}}\label{size-paramter-1dfa}
Let $m$ be a log-space size parameter over alphabet $\Sigma$. If $m$ is polynomially bounded and polynomially honest, there is a nonuniform family $\{M_n\}_{n\in\nat}$ of 1dfa's equipped with write-once output tapes such that each $M_n$ has $n^{O(1)}$ inner states and $M_n$ produces $1^{m(x)}$ on its  output tape from each input $x\in\Sigma^*$.
\end{lemma}

We also consider finite automata whose input-tape heads either move to the right or make stationary moves. Such automata are briefly called \emph{1.5 way}. If we replace 1dfa's in the definition of $\oned$ by 1.5dfa's, then we obtain $\mathrm{1.5D}$. Clearly, $\oned\subseteq \mathrm{1.5D}\subseteq \twod$ follows.
As is shown in the lemma below, $\mathrm{1.5D}$ actually coincides with $\oned$. Similarly, we can consider \emph{1.5-way quantum finite automata with flexible garbage tapes} (or 1.5qfa's).
Nevertheless, as we will show in Lemma \ref{one-vs-one-point-five}, the same equality does not hold for quantum finite automata, exhibiting a clear difference between 1.5dfa's and 1.5qfa's.

\begin{lemma}\label{dfa-one-point-five}
$\oned = 1.5\mathrm{D} \neq \twod$.
\end{lemma}

\begin{proof}
Clearly, $\oned$ is included in $1.5\mathrm{D}$. For the converse inclusion, let $\{M_n\}_{n\in\nat}$ be a nonuniform family of 1.5dfa's. We want to simulate each $M_n$ by a certain 1dfa, say, $N_n$ of $O(n^2)$ states. The desired 1dfa $N_n$ works as follows.
On input $x$, if $M_n$ moves its tape head to the right, then $N_n$ makes the same move.
Consider the case where $M_n$ is in inner state $q$ and makes its tape head stay still.
Let $\delta_n$ denote the transition function of $M_n$. Assume that there are a number $k\geq1$ and a series $p_1,p_2,\ldots,p_k$ of inner states for which $\delta_n(q,\sigma)=(p_1,0)$, $\delta_n(p_i,\lambda)=(p_{i+1},0)$ for any $i\in[k-1]$, and $\delta_n(p_k,\lambda)=(p,1)$.
Since $M_n$ must halt for an arbitrary input $x$, it follows that $k<|Q|$. We define a transition function $\delta'_n$ of $N_n$ as $\delta'_n(q,\sigma)= p$. The obtained $N_n$ is clearly 1-way and simulates $M_n$ on all inputs. Note that $N_n$ uses at most $|Q|$ states. Therefore, we conclude that $1.5\mathrm{D}\subseteq \oned$.

The separation $1.5\mathrm{D}\neq\twod$ follows instantly from the known result of $\oned\neq\twod$ \cite{Kap09}.
\end{proof}

A family $\LL = \{(L_n^{(+)},L_n^{(-)})\}_{n\in\nat}$ of promise decision problems over a fixed alphabet $\Sigma$ is said to have a \emph{polynomial ceiling} if there exists a polynomial $r$ satisfying $L_n^{(+)}\cup L_n^{(-)} \subseteq \Sigma^{\leq r(n)}$ for all $n\in\nat$.
The nonuniform state complexity class $\twod/\poly$ is the restriction of $\twod$, whose families $\LL$ all have appropriate polynomial ceilings. Similarly, we can obtain $\twon/\poly$ from $\twon$.

Given any nonuniform state complexity class $\FF$, $\co\FF$ consists of all families $\LL$ of promise decision problems for which $\co\LL$ belongs to $\FF$.
Using this notion, we obtain, e.g., $\co\onen$ from $\onen$.

\section{One-Way Quantum Finite Automata Families}\label{sec:one-way-QFA}

One-way finite automata are often used to model \emph{data processing} where   streaming input data are processed instantly with little memory space, since their tape heads read input strings from left to right without stopping. Notice that, by our definition of one-wayness, 1-way finite automata halt exactly in $|x|+2$ steps for any given input $x$. For this reason, in many cases, it is possible to demonstrate class separations among nonuniform state complexity classes.

Formally, the notation $\onebq$ denotes the collection of all nonuniform families $\{(L_n^{(+)},L_n^{(-)})\}_{n\in\nat}$ of promise decision problems over fixed input alphabets $\Sigma$ (not depending on $n$ in each family) such that there exist a nonuniform family $\{M_n\}_{n\in\nat}$ of 1qfa's, two polynomials $p$ and $r$, and a constant $\varepsilon\in[0,1/2)$ satisfying the following: for each index $n\in\nat$, (1) for any string  $x\in L_n^{(+)}$, $M_n$ accepts $x$ with probability at least $1-\varepsilon$ and,  for any string $x\in L_n^{(-)}$, $M_n$ rejects $x$ with probability at least $1-\varepsilon$, (2) $M_n$ uses at most $p(n)$ inner states, and (3) $M$'s garbage alphabet has size at most $r(n)$. When $M_n$ satisfies Condition (1), we simply say that $M_n$ \emph{solves} (or \emph{recognizes})  \emph{$(L_n^{(+)},L_n^{(-)})$ with error probability at most $\varepsilon$}. In this case, $M_n$ is said to \emph{make bounded errors}.
In contrast, we obtain $\oneq$ from $\onebq$ if we change Condition (1) to the following new condition for each index $n\in\nat$: (1$'$) for every string $x\in L_n^{(+)}$, $M_n$ accepts $x$ with probability more than $1/2$ and, for every string $x\in L_n^{(-)}$, $M_n$ rejects $x$ with probability at least $1/2$.
In this case, we say that $M_n$ \emph{makes unbounded errors}. Notice that these unbounded-error criteria supersede the aforementioned bounded-error criteria, resulting in the inclusion $\onebq\subseteq \oneq$.
In addition, we also obtain $\onenq$ from $\onebq$ if, instead of Condition (1), we use the following new condition for each $n\in\nat$: (1$''$) for any string $x\in L_n^{(+)}$, $M_n$ accepts $x$ with positive probability, and for any string $x\in L_n^{(-)}$, $M_n$ rejects $x$ with certainty.
Furthermore, we obtain $1.5\mathrm{BQ}$ from $\onebq$ by substituting 1.5qfa's for 1qfa's in the definition of $\onebq$.

In a way similar to introducing $\oneq$, we define $\onep$ using \emph{one-way probabilistic finite automata} (or 1pfa's, for short), whose transition probabilities are \emph{arbitrary numbers} in the real interval $[0,1]$, which make unbounded-error probability. With the use of the bounded-error criteria instead, we further obtain $\onebp$ from $\onep$.
Analogously to $\onebq\subseteq \oneq$, we can obtain $\onebp\subseteq \onep$.
By allowing underlying 1dfa's to use exponentially many inner states, we have already defined $2^{\oned}$ in Section \ref{sec:complexity-class}.
Similarly, we define $2^{\onebq}$, $2^{\oneq}$, $2^{\onebp}$, and $2^{\onep}$,  respectively from $\onebq$, $\oneq$, $\onebp$, and $\onep$.

There are known inclusions and separations: $\oned\subsetneqq \onen \subsetneqq 2^{\oned}$, $\oned=\co\oned$, and $\onen\neq \co\onen$ \cite{Kap09,Kap12}. To obtain Figure \ref{fig:relationship} as our first major contribution, we intend to prove the following collapse and separation relationships among the nonuniform state complexity classes listed in the figure.


\begin{theorem}\label{main-theorem}
\renewcommand{\labelitemi}{$\circ$}
\begin{enumerate}
  \setlength{\topsep}{-2mm}%
  \setlength{\itemsep}{1mm}%
  \setlength{\parskip}{0cm}%

\item $\oned\subsetneqq \onebp \subsetneqq \onebq \subseteq 2^{\oned}$.

\item $\onebq=\co\onebq$ and $\onep=\co\onep$.

\item $\onebq \subsetneqq \mathrm{1.5BQ}$.

\item $\onen \subsetneqq \onep \subseteq \oneq$ and $2^{\oned}\nsubseteq \onep$.

\item $\onebp\subsetneqq 2^{\onebp}$ and $\onep\subsetneqq 2^{\onep}$.

\item $\onen \subsetneqq \onenq \subseteq \oneq$ and $\onenq\nsubseteq 2^{\oned}$.
\end{enumerate}
\end{theorem}

The proof of Theorem \ref{main-theorem} is split into Lemmas \ref{simple-relation}--\ref{one-vs-one-point-five} that follow shortly.

Concerning the state complexity of 1qfa's, unfortunately, the past literature has only a few results. Bianchi, Mereghetti, and Palano \cite{BMP14}, for example, demonstrated an exponential gap between the state complexities of 1dfa's and 1qfa's.
Their result leads to the class separation $\oned\neq \onebq$, which is further strengthened in Lemma \ref{simple-relation}(1).
Let us present the following five relationships, which are partly due to \cite{AF98,AN09,GY15,KW97,YS11}.

\begin{lemma}\label{simple-relation}
(1) $\oned \subsetneqq \onebp$.
(2) $\onebp \subsetneqq \onebq$.
(3) $\onebq\subseteq 2^{\oned}$.
\end{lemma}

\begin{proof}
(1) Since $\oned\subseteq \onebp$ is obvious, we here show that $\oned\neq\onebp$.
Following \cite{GY15}, we define $\mathcal{TRIO} = \{(TRIO_n^{(+)},TRIO_n^{(-)})\}_{n\in\nat}$ with $TRIO_n^{(+)} =\{\#x_1^2y_1\#x_2^2y_2\#\cdots \# x_{n^2}^2y_{n^2}\mid \forall i\in[n^2](x_i,y_i\in\{0,1\}^n), (x_1,x_2,\ldots,x_{n^2})\prec  (y_1,y_2,\ldots,y_{n^2})\}$ and $TRIO_n^{(-)} =\{\#x_1^2y_1\#x_2^2y_2\#\cdots \# x_{n^2}^2y_{n^2}\mid \forall i\in[n^2](x_i,y_i\in\{0,1\}^n), (x_1,x_2,\ldots,x_{n^2})\succ (y_1,y_2,\ldots,y_{n^2})\}$,
where $(x_1,\ldots,x_r)\succ (y_1,\ldots,y_r)$ iff, for any index $i\in[r]$, there exists an index $j\in[n]$ satisfying $(x_i)_{(j)}> (y_i)_{(j)}$.
It was shown in \cite[Theorems 5.1--5.2]{GY15} that $\mathcal{TRIO}$
can be solved by a certain family of $(4n+3)$-state 1pfa's with error probability at most $(1-\frac{1}{n})^{n^2}$ and that any family of 1dfa's  solving $\mathcal{TRIO}$ requires at least $2^n$ inner states.
These results assert that $\mathcal{TRIO}\in\onebp$ but $\mathcal{TRIO}\notin \oned$.

(2) With the use of a flexible garbage tape, any 1dfa can be simulated by an appropriately chosen 1qfa in such a way that, at every step, the 1qfa discards onto its garbage tape the information on the current inner state of the 1dfa.
In a similar way, any 1pfa can be simulated with the same error probability by an appropriate 1qfa. This instantly implies that $\onebp\subseteq\onebq$.

The separation $\onebp\neq \onebq$ is shown as follows. Given two positive integers $m$ and $r$, the notation $m|r$ expresses that $m$ is \emph{divisible} by $r$ (i.e., $r$ divides $m$). Let $e(n)=2^n$ for all $n\in\nat$.
Given any index $n\in\nat^+$, we define $\Sigma_n =\{ a^jb^{m} \mid j,m\in\nat^{+} \wedge j,m<2^{e(n)} \}$.
Let us consider $2MOD_n^{(+)} =\{a^jb^{m} \in \Sigma_n : j|pr(e(n))\}$ and $2MOD_n^{(-)} =\Sigma_n - 2MOD_n^{(+)}$, where $pr(m)$ is the largest prime number $p$ satisfying $p\leq m$ if $m\geq2$, and $pr(1)$ is undefined. We write $2\mathcal{MOD}$ for the family $\{(2MOD_n^{(+)},2MOD_n^{(-)})\}_{n\in\nat}$.
By modifying a 1qfa construction in \cite{AF98,AN09}, we can build a 1qfa that solves $(2MOD_n^{(+)},2MOD_n^{(-)})$ using $O(\log{pr(e(n))})$ inner states with bounded-error probability. Since $2^{n/c}\leq pr(e(n)) \leq 2^n-1$ for a certain absolute constant $c\geq1$, we conclude that  $\Theta(\log{pr(e(n))}) = \Theta(n)$, and thus $2\mathcal{MOD}$ belongs to $\onebq$.

Next, we want to verify that, for any integer $n\geq1$, no 1pfa with less than $pr(e(n))$ inner states solves $(2MOD_n^{(+)},2MOD_n^{(-)})$ with bounded-error probability.
In comparison, given each prime number $p$, let $MOD_{p}$ denote the unary language $\{a^j : j\in\nat^{+}, j|p\}$. It is known in \cite{AF98,AN09} that (*) for any prime number $p$, any 1pfa needs at least $p$ inner states in order to recognize $MOD_{p}$.
If there is a 1pfa $M$ with $k$ states ($k<pr(e(n))$) that solves $(2MOD_n^{(+)},2MOD_n^{(-)})$ with bounded-error probability, then we can convert $M$ to another 1pfa that can recognize $MOD_{pr(e(n))}$ with bounded-error probability by setting $m=0$ in the definition of $2MOD_n^{(+)}$. This is a contradiction against the above statement (*).
Therefore, $2\mathcal{MOD}$ cannot belong to $\onebp$.

(3) We start with the following general claim. Kondacs and Watrous \cite{KW97} proved a similar claim for 1qfa's with no garbage tape.
In the claim, we need to deal with flexible garbage tapes as well.

\begin{claim}\label{garbage-tape-1dfa}
Any $n$-state bounded-error 1qfa with a garbage alphabet of arbitrary size  can be simulated by a certain $2^{O(n^2)}$-state 1dfa.
\end{claim}

\begin{proof}
Let $\Sigma$ be any alphabet and let $\varepsilon\in[0,1/2)$ be
any error bound.
Take an $n$-state 1qfa $M$ with a garbage alphabet $\Xi$ of size $r\in\nat^{+}$  with error probability at most $\varepsilon$.
Following an argument similar to \cite{KW97} (or \cite{AY15}),  we
want to simulate $M$ classically.
Let $Q$ denote the set of all inner states of $M$ with $n=|Q|$ and
take two Hilbert spaces $\HH_{Q} = span \{\qubit{q}\mid q\in Q\}$ and $\HH_{\Xi} = span\{\qubit{z}\mid z\in\Xi^{\leq |x|+2}\}$.
Let $x$ be any input in $\Sigma^*$.
We need to express the content of the garbage tape as a string in $\Xi^{\leq |x|+2}$. For convenience, we write $A$ for the set $\Xi^{\leq |x|+2}$. We then define the unit sphere $\SSS_n$ to be $\{\qubit{\phi}\in \HH_{Q}\otimes \HH_{\Xi}\mid \|\qubit{\phi}\|=1\}$.

Let us consider any superposition $\qubit{\phi}$ of surface configurations of $M$ on the input $x$.  Assume that $\qubit{\phi}$ is of the form $\sum_{z\in A}\sum_{q\in Q}\alpha_{q,z}\qubit{q}\qubit{z}$ for appropriate amplitudes $\alpha_{q,z}$'s.
For each fixed string $z\in A$, let $\qubit{\phi_z} = (1/p_z)\sum_{q\in Q}\alpha_{q,z}\qubit{q}$, where $p_z$ is a normalizing positive constant.  We are focused on the density operator $\rho = \sum_{z\in A}p_z\density{\phi_z}{\phi_z}$. Since $\rho$ is of dimension $n$, we can choose $k$ vectors $\{\qubit{\phi'_i}\}_{i\in[k]}$ and $k$ numbers $\{p'_i\}_{i\in[k]}$ for which $\rho$ is  expressed as $\sum_{i\in[n]} p'_i\density{\phi'_i}{\phi'_i}$.

A set $S$ of vectors in $\HH_{Q}$ is called an \emph{$\varepsilon$-net} if, for any vector $\qubit{\phi}$ in $\HH_{Q}$, there always exists another  vector $\qubit{\psi}$ in $S$ satisfying
$\|\qubit{\phi}-\qubit{\psi}\|\leq \varepsilon$.
The Solovay-Kitaev theorem (e.g., \cite{KSV02,NC00}) then ensures the existence of an $\varepsilon$-net of $2^{O(n)}$ vectors. For the Hilbert space $\HH_{Q}^{\otimes{n}}$, there is an $\varepsilon$-net, say, $S$  of $(2^{O(n)})^n$ ($=2^{O(n^2)}$) vectors.
Treating each vector in $S$ as a new inner state, we can construct a new 1dfa that recognizes $L$. Note that the state complexity of this 1dfa is at most $|S|$. The lemma follows from the fact that $|S|$ is bounded by $2^{O(n^2)}$.
\end{proof}

The desired inclusion $\onebq\subseteq 2^{\oned}$ comes directly from Claim \ref{garbage-tape-1dfa}. This completes the proof.
\end{proof}

\begin{lemma}\label{oneq-topics}
(1) $\onep \subseteq \oneq$.
(2) $\onenq\subseteq \oneq$.
\end{lemma}

\begin{proof}
(1) Similarly to the proof of Lemma \ref{simple-relation}(2) for $\onebp\subseteq \onebq$, the use of a garbage tape makes it possible to simulate any unbounded-error 1pfa by an appropriately designed unbounded-error 1qfa, and therefore $\onep\subseteq \oneq$ follows.

(2) To verify $\onenq\subseteq \oneq$, for any given 1qfa $M$, it suffices to construct another 1qfa $N$ so that $N$ splits its first move into two parts:
(i) with probability $1/2$, $N$ starts simulating $M$ on $x$ and, with probability $1/2$, $N$ immediately accepts $x$ .
Let us consider the case where $N$ tries to simulate $M$.
If $M$ accepts $x$ with positive probability $\alpha$, then $N$ also accepts $x$ with probability $\frac{1}{2}+\frac{\alpha}{2}$. On the contrary, if $M$ rejects $x$ with certainty, then $N$ rejects $x$ with probability $\frac{1}{2}$. Thus, $N$ correctly simulates $M$ with unbounded-error probability.
\end{proof}


\begin{lemma}\label{BQ-one-vs-exponential}
(1) $\onebq = \co\onebq$.
(2) $\onebp\subsetneqq 2^{\onebp}$.
\end{lemma}

\begin{proof}
(1) This is shown simply by exchanging the roles of accepting states and rejecting states of underlying 1qfa's.

(2) Obviously, $\onebp\subseteq 2^{\onebp}$ follows.
If $\onebp=2^{\onebp}$ holds, then $2^{\onebp}\subseteq \onep$ follows instantly because of $\onebp\subseteq \onep$.
Since $\oned\subseteq \onebp$, we conclude from $2^{\onebp} \subseteq\onep$ that $2^{\oned}\subseteq \onep$. However, this contradicts Lemma \ref{unbounded-error-case}(1).
Therefore, we obtain $\onebp\neq 2^{\onebp}$.
\end{proof}

\begin{lemma}\label{unbounded-error-case}
(1) $2^{\oned}\nsubseteq \onep$.
(2) $\onep = \co\onep$.
(3) $\onen\subsetneqq \onep$.
(4) $\onep\subsetneqq 2^{\onep}$.
\end{lemma}

\begin{proof}
(1) Let us consider $L_{NH} = \{0^x10^{y_1}1\cdots 10^{y_k}1\mid x,y_1,\ldots,y_k\in\nat^{+},\exists l\in[k]\;(x=\sum_{i\in[l]} y_i)\}$. Associated with $L_{NH}$, we define the family $\LL_{NH} = \{(L_n^{(+)},L_n^{(-)})\}_{n\in\nat}$ by setting $L_n^{(+)} =\{ z\in L_{NH} \mid |z|\leq n^{\log{n}}\}$ and $L_n^{(-)} = \{z\notin L_{NH}\mid |z|\leq n^{\log{n}}\}$.
We want to show that $\LL_{NH}\in 2^{\oned}-\onep$.
As demonstrated by Freivalds and Karpinski \cite{FK94}, no one-way PTM recognizes $L_{NH}$ using $o(\log{n})$ space with unbounded-error probability. This means that no 1pfa with $o(n^{\log{n}})$ inner states can recognize $(L_n^{(+)},L_n^{(-)})$ for almost all $n\in\nat$.
We thus conclude that $\LL_{NH}$ does not belong to $\onep$. However, as shown in \cite{FK94}, $\LL_{NH}$ can be solved deterministically using $O(n^{\log{n}})$ ($\subseteq O(2^n)$) inner states, and thus $\LL_{NH}$ falls in $2^{\oned}$.

(2) Firstly, we modify the acceptance/rejection criteria of unbounded-error 1pfa's.

\begin{claim}
Given a 1pfa $M$, there exists another 1pfa $N$ such that, for any $x$, if $M$ accepts $x$ with probability more than $1/2$ (resp., rejects $x$ with probability at least $1/2$), then $N$ accepts (resp., rejects) $x$ with probability more than $1/2$.
\end{claim}

Let us show how the desired result follows from this claim.
By the claim, we first pick up a nonuniform family $\{M_n\}_{n\in\nat}$ of polynomial-size unbounded-error 1pfa's that always either accept or reject with probability  more than $1/2$. We then simply exchange the roles of accepting and rejecting states.
This concludes that $\onep=\co\onep$.

To show the claim, let $\alpha_0$  be the minimum nonzero transition probability of $M$. Without loss of generality, we assume that $\alpha_0\leq \frac{1}{2}$.
Note that the acceptance probability of $M$ on $x$ (if any) is more than $1/2+\alpha_0^{|x|+2}$ since $M$ halts within $|x|+2$ steps.
We design a new 2pfa $N$ to behave as follows.
On input $x$, in scanning $\cent$, with probability $1-\alpha_0/2$, we mimic the behavior of $M$. With probability $\alpha_0/2$, from the initial inner state $q_0$, we enter a new inner state, say, $\tilde{q}_0$.
Whenever we read each tape symbol, from the current inner state $\tilde{q}_0$, we re-enter $\tilde{q}_0$ with probability $\alpha_0$.
Furthermore, with equal probability $\frac{1}{2}(1-\alpha_0)$, we enter both an accepting state and a rejecting state.
When we read $\dollar$, from $\tilde{q}_0$, we enter a rejecting state with certainty.
It thus follows that either the accepting probability of $N$ is more than $\frac{1}{2}+\frac{5}{8}\alpha_0^{|x|+2}>\frac{1}{2}$ or the rejecting probability of $N$ is more than $\frac{1}{2}+\frac{1}{4}\alpha_0^{|x|+2}>\frac{1}{2}$.

(3) The inclusion $\onen\subseteq \onep$ can be shown by modifying underlying 1nfa's as follows: choose all nondeterministic transitions ``probabilistically'' and enter both an accepting state and a rejecting state with equal probability only along each rejecting path.
The separation  $\onen\neq \onep$ instantly follows because $\onep =\co\onep$ by (2) and $\onen\neq \co\onen$ by  \cite{Kap09,Kap12}.

(4) Clearly, it follows that $\onep\subseteq 2^{\onep}$. If $\onep=2^{\onep}$, then we obtain $2^{\oned}\subseteq 2^{\onep}=\onep$ since $\oned\subseteq\onep$. As a consequence,  $2^{\oned}\subseteq \onep$ follows. This is a clear contradiction against Lemma \ref{unbounded-error-case}(1). Therefore, $2^{\onep}$ is different from $\onep$.
\end{proof}


\begin{lemma}
$\onen\subsetneqq\onenq \nsubseteq 2^{\oned}$.
\end{lemma}

\begin{proof}
It is easy to show that $\onen\subseteq \onenq$ with the use of the amplitude set $\{0,\pm1/2,\pm1\}$.
To lead to $\onen\neq\onenq$, it suffices to prove that $\onenq\nsubseteq 2^{\oned}$ because, if $\onen=\onenq$, then  we obtain $\onenq\subseteq 2^{\oned}$ from $\onen\subseteq2^{\oned}$ \cite{Kap09}, leading to a contradiction.

Our goal is now to prove that $\onenq\nsubseteq 2^{\oned}$.
Let $\Sigma=\{0,1\}$ and set $e(n)=2^n$ for any $n\in\nat$.
Consider the family $\NEQ = \{(NEQ_n^{(+)},NEQ_n^{(-)})\}_{n\in\nat}$, where $NEQ_n^{(+)} =\{w\in\Sigma^{2^{e(n)} }\mid \#_0(w)\neq\#_1(w)\}$ and $NEQ_n^{(-)} =\Sigma^{2^{e(n)}} - NEQ_n^{(+)}$ for each index $n\in\nat$, where $\#_b(x)$ expresses the total number of symbol $b$ in string $x$.
Since $\NEQ$ belongs to $\onenq$ \cite{BC01,BP02}, we want to claim that $\NEQ\notin2^{\oned}$. Assume otherwise.
Note that $\co\NEQ$ also belongs to $2^{\oned}$.  By a communication-complexity argument (e.g., \cite{KN06}) or a swapping lemma \cite{Yam08}, it is not difficult to show that $N_n$ requires at least $o(2^{e(n)})$ inner states to recognize $NEQ_n^{(-)}$ in the worst case. This is a clear contradiction against our assumption of $\NEQ\in 2^{\oned}$.
\end{proof}


To finish the proof of Theorem \ref{main-theorem}, we still need to examine the computational power of 1.5-way finite automata. With the use of 1.5qfa's, we can define the associated nonuniform state complexity class $1.5\mathrm{BQ}$. We have already seen in Lemma \ref{dfa-one-point-five} that $\mathrm{1.5D}$ coincides with $\oned$. However, as shown in Lemma \ref{one-vs-one-point-five}, the same coincidence does not occur between $1.5\mathrm{BQ}$ and $\onebq$.

\begin{lemma}\label{one-vs-one-point-five}
$\onebq \neq \mathrm{1.5BQ}$.
\end{lemma}

\begin{proof}
To lead to a contradiction, we assume that $\onebq=\mathrm{1.5BQ}$. Let us consider the family $\mathcal{EQ} =\{(EQ_n^{(+)},EQ_n^{(-)})\}_{n\in\nat}$, each element $(EQ_n^{(+)},EQ_n^{(-)})$ of which is defined as $EQ_n^{(+)} =\{a^{\bar{n}}b^{\bar{n}}\mid \bar{n} = 2^{e(n)}\}$ and $EQ_n^{(-)} = \{a^ib^j\mid i,j\geq0, i\neq j, i+j=2^{e(n)+1}\}$ for each index $n\in\nat$, where $e(n)=2^n$.
Since we need to deal only with \emph{valid} inputs given to $\mathcal{EQ}$, for each promise decision problem $(EQ_n^{(+)},EQ_n^{(-)})$, we do not need to check that any given input $x$ has length exactly $e(n)+1$.
As noted in \cite{AY15}, each $(EQ_n^{(+)},EQ_n^{(-)})$ can be solved by an appropriately designed bounded-error 1.5qfa with a constant number of inner states, and thus $\mathcal{EQ}$ belongs to $\onebq$ by our assumption.
Since $\onebq\subseteq 2^{\oned}$ by Lemma \ref{simple-relation}(3),
there are a polynomial $p$ and a nonuniform family of 1dfa's that solves $\mathcal{EQ}$ using at most $2^{p(n)}$ inner states. However, any 1dfa solving $(EQ_n^{(+)},EQ_n^{(-)})$ requires at least $2^{e(n)}$ inner states. Since $p(n)<e(n)$ for almost all $n\in\nat$, we obtain a clear contradiction. Therefore, $1.5\mathrm{BQ}$ differs from $\onebq$.
\end{proof}


Recall that our 1qfa uses a \emph{flexible garbage tape}. When both an input-tape head and a garbage-tape head synchronize in their movement (thus, the garbage-tape head makes no stationary move), we can prove another close relationship between 1pfa's and 1qfa's in Lemma \ref{rigid-garbage-tape}.
Such a restrictive use of the garbage tape is referred to as a \emph{rigid garbage tape} in contrast to the  flexible garbage tape.
A similar model was already seen in \cite[Section 5.2]{Yam14b}, in which a rewritable track of an input tape is used as our rigid garbage tape.
To emphasize the use of rigid garbage tapes, we use the special notation $\oneq^{(+)}$ opposed to $\oneq$ when underlying 1qfa's use rigid garbage tapes.


\begin{lemma}\label{rigid-garbage-tape}
(1) $\oneq^{(+)} = \onep$. (2) $2^{\oneq^{(+)}} = 2^{\onep}$.
\end{lemma}

\begin{proof}
(1)
The proof of Lemma \ref{oneq-topics}(2) has actually proven that $\onep\subseteq \oneq^{(+)}$.

Our proof of $\oneq^{(+)} \subseteq \onep$ essentially follows an argument  used in \cite{TYL10} together with \cite{Yam03}.
In \cite[Lemma 8.1]{TYL10}, for any linear-time one-tape QTM with rational amplitudes, its acceptance probability on each input $x$ is calculated in terms of two acceptance probabilities of liner-time one-tape PTMs.
In \cite[Theorem 7.1]{Yam03}, a language of the form $\{x\mid p_{acc,M_1}(x)> p_{acc,M_2}(x)\}$ for two polynomial-time QTMs $M_1$ and $M_2$ belongs to $\pp$. It is also possible to restrict all transition amplitudes of any QTM on real numbers.
This restriction simplifies our construction of the desired 1pfa.
A major deviation from \cite{TYL10}, however, is that we allow 1qfa's to take arbitrary real transition amplitudes in the real interval $[0,1]$.

Let us consider an arbitrary  family $\LL=\{(L_n^{(+)},L_n^{(-)})\}_{n\in\nat}$ of promise decision problems in $\oneq^{(+)}$ and also a nonuniform family $\{M_n\}_{n\in\nat}$ of  $r(n)$-state 1qfa's with rigid garbage tapes that recognizes $\LL$ \emph{in exactly $|x|+2$ steps}  with bounded-error probability, where $r$ is a suitable polynomial.
Let $\alpha_n$ denote the smallest absolute value of any transition amplitude used by $M_n$.

In a natural way, we express each computation path of $M_{n}$ on input $x$ as strings of length $|x|+2$ over an appropriate alphabet, say, $\Theta$. To such a computation path $y$, we assign the product of all transition amplitudes taken by $M_{n}$ along the path $y$. This product is briefly called the \emph{amplitude associated with $y$}.
Let us denote by $amp(z,y,z)$ the amplitudes associated with a computation path $y$ of $M_{n}$ on $x$ ending with a final surface configuration $z$.

Firstly, we wish to define two functions. The first function $f_{+}(x,z)$ is set to be $\sum^{(+)}_{y} |amp(x,y,z)|$, where $y$ ranges over all accepting computation paths of $M_{n}$ on $x$ ending with $z$  for which $amp(x,y,z)$ is a positive real number. We also define $f_{-}(x,z) =\sum^{(-)}_{y} |amp(x,y,z)|$ similarly except that we collect $amp(x,y,z)$'s  having negative values. The acceptance probability $p_{n,acc}(x)$ of $M_{n}$ on $x$ is therefore expressed as  $\sum_{z}(f_{+}(x,z) - f_{-}(x,z))^2$.
To compute this value $p_{n,acc}(x)$, we need to compute two additional values $p_{n}^{(+)}(x) = \frac{1}{2} \sum_{z} (f_{+}^2(x,z) + f_{-}^2(x,z))$ and $p_{n}^{(-)}(x) = \sum_{z}f_{+}(x,z)f_{-}(x,z)$.
Since $M_{n}$ makes bounded-error probability, for an appropriate constant $\varepsilon\in[0,1/2)$, it follows that (i) $p_{n,acc}(x) \geq 1-\varepsilon$ iff $p^{(+)}_{n}(x) > p^{(-)}_{n}(x)$ and (ii) $p_{n,acc}(x) \leq \varepsilon$ iff $p^{(+)}_{n}(x) <  p^{(-)}_{n}(x)$.

Secondly, we want to produce in a probabilistic manner the values  $p^{(+)}_{n}(x)$ and $p^{(-)}_{n}(x)$ multiplied by $\zeta_{n,x}$,
where $\zeta_{n,x} = \alpha_0^{2(|x|+2)}\cdot |\Theta|^{-2(|x|+2)}$.
Since the computation paths of $M_n$ on $x$ may interfere with one another, we need to ignore such an interference and classically generate these computation paths with certain probabilities.
Concerning $p^{(-)}_{n}(x)$, we want to produce $\zeta_{n,x}\sum_{y_1}\sum_{y_2} |amp(x,y_1,z)||amp(x,y_2,z)|$, provided that $y_1$ makes $amp(x,y_1,z)$ positive and $y_2$ makes $amp(x,y_2,z)$ negative.
For this purpose, we first generate all pairs $(y_1,y_2)$ in $\Theta^{|x|+2}\times \Theta^{|x|+2}$ at random, simulate $M_n$ along two computation paths $y_1$ and $y_2$ simultaneously, check that the same final configuration is reached, check that $y_1$ and $y_2$ satisfy the aforementioned condition, and finally accept with probability $\zeta_{n,x} |amp(x,y_1,z)||amp(x,y_2,z)|$. Similarly in essence, we can also estimate $p^{(+)}_{n}(x)$.
The simultaneous simulation of the two computation paths $y_1$ and $y_2$ is possible because the garbage-tape head of $M_n$ always
moves to the right.

In a way similar to  \cite[Theorem 7.1]{Yam03},
we can construct an $n^{O(1)}$-state unbounded-error 1pfa whose acceptance probability on the input $x$ equals $\zeta_{n,x} p^{(+)}_{n,acc}(x)$ and also another 1pfa whose acceptance probability is  $\zeta_{n,x} p^{(-)}_{n,acc}(x)$. Finally, we combine those two 2pfa's to a new 1pfa $D_n$ so that (i$'$) $p^{(+)}_{n}(x) > p^{(-)}_{n}(x)$ iff $D_n$ accepts $x$ with probability more than $1/2$ and (ii$'$) $p^{(+)}_{n}(x) < p^{(-)}_{n}(x)$ iff $D_n$ accepts $x$ with probability less than $1/2$. Hence, $D_n$ solves $(L_n^{(+)},L_n^{(-)})$ with unbounded-error probability.
This concludes that $\LL$ belongs to $\onep$.

(2) An argument similar to (1) leads to the collapse $2^{\oneq^{(+)}} = 2^{\onep}$.
\end{proof}

\section{Two-Way Quantum Finite Automata Families}\label{sec:two-way-QFA}

We have already introduced $\onebq$ and $\oneq$ using 1qfa's in Section \ref{sec:one-way-QFA}.
Furthermore, let us define $\twobq$ to be the collection of all nonuniform families $\LL= \{(L_n^{(+)},L_n^{(-)})\}_{n\in\nat}$ of promise decision problems such that, for each family $\LL$, there exist a nonuniform family $\{M_n\}_{n\in\nat}$ of 2qfa's, two polynomials $p$ and $r$, and an error bound $\varepsilon\in[0,1/2)$ satisfying the following: for each index $n\in\nat$, (1)  $M_n$ makes error probability at most $\varepsilon$ on all inputs in $\Sigma_n=L_n^{(+)}\cup L_n^{(-)}$,
(2) $M_n$ uses at most $p(n)$ inner states, and (3) $M_n$'s garbage alphabet has size at most $r(n)$.
We further define $\twoq$ similarly to $\twobq$ using unbounded-error probability instead of bounded-error probability.

In a similar fashion, we define $\twobp$ in terms of bounded-error 2pfa's having polynomially many states.
The unbounded-error analogue of $\twobp$ is denoted by $\twop$.\footnote{In \cite{Kap09}, the polynomial-time $\twobp$ was considered under the name of $\mathrm{2P}_2$ and the polynomial-time $\twop$ was studied under the name of $\mathrm{2P}$ but they are restricted to so-called ``regular'' language families. Here, we demand neither the polynomial time-bound nor the regular family requirement.}


Notice that, for any suffix $\AAA\in\{\mathrm{D},\mathrm{N}, \mathrm{BP},\mathrm{BQ}\}$, we obtain an inclusion $1\AAA \subseteq 2\AAA$ since ``1-way'' machines can be thought as a special case of ``2-way'' machines.
Unlike the one-way case, there are few known separations among 2-way nonuniform state complexity classes, except that
$\oned\subsetneqq \twod\subseteq \twon \subsetneqq 2^{\oned}$ \cite{Kap09}.
Even for the basic classes $\twod$ and $\twon$, unfortunately, we do not know that $\twon\neq \co\twon$ as well as $\onen\nsubseteq \twod$.
For the probabilistic counterpart, nonetheless, we can draw the following  relationships from an early work of Dwork and Stockmeyer \cite{DS90}.

\begin{lemma}\label{Dwork-Stockmeyer}
(1) $\twobp \subseteq 2^{\oned}$. (2) $\twobp \nsubseteq\twon$.
\end{lemma}

\begin{proof}
(1) From \cite[Theorem 6.1]{DS90}, we can conclude that any $n$-state 2pfa with error probability $\varepsilon\in[0,1/2)$ running in expected $O(n^k)$ time for any fixed constant $k>0$ can be precisely simulated by
an appropriately chosen 1dfa of $2^{O(n^2\log{n})}$ states. This implies the first inclusion.

(2) Let us define an example family $\LL=\{(L_n^{(+)},L_n^{(-)})\}_{n\in\nat}$ of promise decision problems as follows: for any index $n\in\nat$, let $L_n^{(+)}=\{0^{e(n)}\}$ and $L_n^{(-)} = \{0^i\mid 0\leq i\leq e(n)^2, i\neq e(n)\}$, where $e(n)=2^n$. It was proven in \cite[Theorem 6.2]{DS90} that (1) for any $n\in\nat$, $(L_n^{(+)},L_n^{(-)})$ is solved with bounded-error probability by a certain 2pfa of $O(\log^2{e(n)}/\log\log{e(n)})$ inner states in expected $O(|x|^2)$ time, where $x$ is a ``symbolic'' input, and (2) any 2nfa requires at least $e(n)$ inner states to solve $(L_n^{(+)},L_n^{(-)})$. These results (1)--(2) help us conclude that $\twobp \nsubseteq\twon$.
\end{proof}


Unlike 1qfa's, it is possible to reduce the garbage alphabets of 2qfa's to  constant size. This reduction in size makes it easier to estimate the computational complexity of 2qfa's.

\begin{lemma}\label{garbage-alphabet}
Given any $n$-state 2qfa $M$ with a garbage alphabet of size $k$, there exists another $O(nk\log{k})$-state 2qfa $N$ with a constant-size garbage alphabet that can  simulate $M$ with the same error probability.
\end{lemma}

\begin{proof}
Let $M =(Q,\Sigma,\{\cent,\dollar\},\Xi,\delta,q_0,Q_{acc},Q_{rej})$ denote any   given $n$-state 2qfa. For simplicity, let $r=\ceilings{\log{k}}$ and assume that $\Xi = \{0,1\}^r$ for a certain $r\in\nat^{+}$.
For each symbol $\xi$ in $\Xi$, we express it as $\xi_1\xi_2\cdots \xi_r$ using $r$ binary symbols. We define a new 2qfa $N$ with $Q'= Q\times \Xi\times [r]$ and $\Xi'=\{0,1\}$ as follows.
The initial state of $N$ is set to be $(q_0,\lambda,0)$.
Assume that $N$ is in inner state $(q,\lambda,0)$.
If $M$ applies a quantum transition of the form $\delta(q,\sigma|p,d,\xi)=\alpha$ with $\alpha\neq0$, then $N$ first writes $\xi_1$ on a garbage tape, enters the next inner state $(p,\xi,1)$ with amplitude $\alpha$, and moves its tape head in direction $d$. For each index $i\in[r-1]$, with no tape-head move, $N$ enters $(p,\xi,i)$ from $(p,\xi,i-1)$ by dumping $\xi_i$ onto the garbage tape. Finally, $N$ enters $(p,\lambda,0)$  from $(p,\xi,r)$ with no tape-head move. It is not difficult to show that $N$ correctly simulates $M$ with the same error probability.
\end{proof}

As a consequence of Lemma \ref{garbage-alphabet}, in the rest of this exposition, we consider only 2qfa's that have garbage alphabets of \emph{constant size} unless otherwise stated.


To complete Figure \ref{fig:relationship}, we further need to show the following relationships.

\begin{lemma}\label{first-claim}
(1) $\twod\subsetneqq \twobp$.
(2) $\twobp \subseteq \twobq$.
(3) $\twobq\subseteq \twoq$.
(4) $\twon\subseteq \twop$.
(5) $\twop \subseteq \twoq$.
\end{lemma}

\begin{proof}
(1) Since every 2dfa can be viewed as a special case of bounded-error 2pfa, we obtain $\twod\subseteq \twobp$. The separation $\twod\neq \twobp$ comes from $\twobp\nsubseteq\twon$ of Lemma \ref{Dwork-Stockmeyer}(2).

(2) Any 2pfa can be simulated by an appropriately designed 2qfa, which  discards all information on the previously taken inner states to a garbage tape. The inclusion $\twobp\subseteq \twobq$ follows immediately.

(3) This is trivial by the difference between the error bound criteria of $\twobq$ and $\twoq$.

(4) In a way similar to Lemma \ref{unbounded-error-case}(3), we can show that $\twon\subseteq \twoq$.

(5) We have already proven $\onep\subseteq \oneq$ in Lemma \ref{simple-relation}(4).
A similar argument proves $\twop\subseteq \twoq$ as well.
\end{proof}

For two-way head moves, the behaviors of two-way finite automata vary significantly depending on their machine types.
For deterministic and nondeterministic computations, as discussed in Section \ref{sec:main-contribution}, the length of ``accepting'' computation paths of 2dfa's and 2nfa's are always upper-bounded \emph{linearly} in input size.
For probabilistic computation, nonetheless, this is not always true. As  demonstrated by Freivalds \cite{Fre81}, bounded-error 2pfa's with no runtime bound in general have more computational power than bounded-error expected-polynomial-time 2pfa's.
Since we are interested in two-way machines running \emph{in expected polynomial time}, we introduce the nonuniform state complexity class $\pt\twobq$ from $\twobq$ by further
requiring each family $\{M_n\}_{n\in\nat}$ of underlying 2qfa's $M_n$ to run in expected $p(n,|x|)$ time on all inputs $x$ for an appropriate polynomial $p$ (depending only on the choice of the family $\{M_n\}_{n\in\nat}$). Similarly, we obtain $\pt\twoq$, $\pt\twobp$, and $\pt\twop$ from $\twoq$, $\twobp$, and $\twop$, respectively.
A similar statement to Lemma \ref{first-claim} also holds for $\pt\twobp$, $\pt\twop$, $\pt\twobq$, and $\pt\twoq$.

\section{Advised QTMs and Nonuniform Families of Quantum Finite Automata}\label{sec:advised-QTM}

The nonuniform state complexity classes $\twod/\poly$ and $\twon/\poly$ are composed of only families of promise decision problems having \emph{polynomial ceilings}. In a similar way, we define two more classes $\twobp/\poly$ and $\twobq/\poly$ using the notion of polynomial ceiling.
From these complexity classes, we can obtain $\pt\twobp/\poly$ and $\pt\twobq/\poly$ by forcing underlying machines to halt in expected polynomial time.

Concerning the acceptance/rejection criteria for families of promise decision problems, quantum and probabilistic computations are quite different from deterministic and nondeterministic computations. Recall from Section \ref{sec:parameter-promise}  the notion of \emph{$\dl$-good} collection $\FF$ of families of promise decision problems.
It is not difficult to verify that $\twod$ and $\twon$ are $\dl$-good because 2dfa's and 2nfa's can be easily forced to either accept or reject on all possible inputs. On the contrary, $\twobq$ and $\twobp$ may not be $\dl$-good. For this very reason, we introduce the special notation $\twobq^{\dagger}$ (resp., $\twobp^{\dagger}$) to denote the largest subclass of $\twobq$ (resp., $\twobp$) that is \emph{closed downward under $\dl$-good extensions}
(i.e., if $\LL_1$ is an $\dl$-good extension of $\LL_2$ and $\LL_2\in\twobq^{\dagger}$, then $\LL_1$ belongs to $\twobq^{\dagger}$).

The sole purpose of this section is to verify Theorem \ref{general-theorem}, from which Corollary \ref{NL-equivalence} follows immediately.
For the proof of the theorem, as a supportive statement inspired by \cite{Yam18}, we first establish a close connection between nonuniform state complexity classes and parameterized complexity classes.
The actual proof of Theorem \ref{general-theorem} will be given in Section \ref{sec:proof-of-theorem}.

\subsection{The Roles of Advice and the Honesty Condition}\label{sec:role-of-advice}

When we discuss quantum and probabilistic finite automata with two-way head moves, the runtime bounds of 2qfa's and 2pfa's are quite essential because, as Freivalds \cite{Fre81} demonstrated, expected-polynomial-time bounded-error 2pfa's are in general less powerful than, say,  expected-exponential-time bounded-error 2pfa's.

Let us begin with the precise statement of Proposition \ref{characterize-PBQL/poly}, in which we give a nice characterization of parameterized decision problems solvable by  advised QTMs in expected polynomial time using logarithmic space in terms of expected-polynomial-time 2qfa's having polynomially many inner  states.
Firstly, we recall that $\phsp$ consists of all parameterized decision problems with polynomially-honest size parameters.

\begin{proposition}\label{characterize-PBQL/poly}
Let $(\AAA,\BB)\in\{(\mathrm{D},\dl), (\mathrm{N},\nl), (\mathrm{BP},\bpl), (\mathrm{BQ},\bql)\}$ and let $\LL=\{(L_n^{(+)},L_n^{(-)})\}_{n\in\nat}$ denote any family of promise decision problems.
\renewcommand{\labelitemi}{$\circ$}
\begin{enumerate}\vs{-2}
  \setlength{\topsep}{-2mm}%
  \setlength{\itemsep}{1mm}%
  \setlength{\parskip}{0cm}%

\item Given each parameterized decision problem $(L,m)$, if $\LL$ is  induced from $(L,m)$, then it follows that $(L,m)\in\para\pt\BB/\poly\cap \phsp$ iff $\LL\in \pt2\AAA/\poly$.

\item If $\LL$ is $\dl$-good and $(K,m)$ is the parameterized decision problem induced from $\LL$, then it follows that $(K,m)\in\para\pt\BB/\poly\cap \phsp$ iff $\LL\in \pt2\AAA/\poly$.
\end{enumerate}
\end{proposition}

Hereafter, we intend to verify Proposition \ref{characterize-PBQL/poly} only for the case of $\AAA=\mathrm{BQ}$ and $\BB=\bql$ since the other cases can be proven in a similar way with only necessary modification. For readability, the proof of Proposition \ref{characterize-PBQL/poly}(1) is split into two lemmas, Lemmas \ref{QTM-to-QFA} and \ref{QFA-to-QTM}.
Lemma \ref{QTM-to-QFA}, which is in fact more general than what the proposition actually needs, states that we can simulate an advised QTM by a certain nonuniform family of 2qfa's of appropriate state complexity.

\begin{lemma}\label{QTM-to-QFA}
Let $m$ be a size parameter and
let $M$ be an advised QTM running with an advice function $h$. Let $r$ be a function on $\nat$ satisfying $|h(n)|\leq r(n)$ for all $n\in\nat$.
For two functions $p$ and $\ell$, there always exists a nonuniform family $\{N_{n,l}\}_{n,l\in\nat}$ of 2qfa's with $O(r(l))\cdot 2^{O(\ell(n,l))}$ inner states such that, for any input $x$, if $M$ accepts (resp., rejects) $(x,h(|x|))$ with bounded-error probability in expected $p(m(x),|x|)$ time using at most $\ell(m(x),|x|)$ space, then $N_{m(x),|x|}$ accepts (resp., rejects) $x$ with bounded-error probability in expected $p(m(x),|x|)$ time.
\end{lemma}

\begin{proof}
Let $p$, $\ell$, and $r$ denote respectively a time-bounding function, a space-bounding function, and an advice-bounding function. Let $m$ be a given size parameter over an input alphabet $\Sigma$ and let $h$ be an advice function from $\nat$ to $\Theta^*$ for a certain advice alphabet $\Theta$ satisfying  $|h(n)|\leq r(n)$ for all $n\in\nat$.
We take an arbitrary advised QTM  $M=(Q,\Sigma,\{\cent,\dollar\},\Gamma, \Theta, \Xi, \delta,q_0,Q_{acc},Q_{rej})$ equipped with
a work alphabet $\Gamma$ and a garbage alphabet $\Xi$.
Assume that, with the use of $h$, $M$ runs in expected $p(m(x),|x|)$ time using at most $\ell(m(x),|x|)$ space on any input $x\in\Sigma^*$. In what follows, we fix $x$ and set $a=h(|x|)$.
Note that, since $p(m(x),|x|)$ is an expected runtime bound of $M$ on $x$, for an appropriately chosen absolute constant $c\geq1$, it is enough to consider only the first $cp(m(x),|x|)$ steps of any computation path of $M$ on $x$ without losing $M$'s bounded-error probability criteria. Remember  that the contents of the input tape and the advice tape do not change during any computation.
A \emph{surface configuration} of $M$ on $x$ is of the form $(q,j,k,y,t,z)$, which indicates that $M$ with garbage-tape content $z$ is in inner state $q$, scanning the $j$th cell of an input tape, the $k$th cell of a work tape containing $y$, and the $t$th cell of an advice tape.

A basic idea is to maintain the partial information on the current surface configuration of $M$ on $x$ in the form of an inner state of the desired 2qfa.
The desired nonuniform family $\NN = \{N_{n,l}\}_{n,l\in\nat}$ of 2qfa's thus has the following form. Given an instance $x$, let $n=m(x)$ and set $N_{n,|x|} = (Q',\Sigma,\{\cent,\dollar\}, \Xi, \delta', q'_{0},Q'_{qcc},Q'_{rej})$. The set $Q'$ consists of all inner states of the form $(q,k,y,t)$, where $q\in Q$, $k\in[0,\ell(n,|x|)+1]_{\integer}$, $y\in\Gamma^*$, and $t\in[0,|a|+1]_{\integer}$.
It then follows that $|Q'|\leq |Q|(\ell(n,|x|)+1)(r(|x|)+2)|\Gamma|^{\ell(n,|x|)+1}$ since $a$ is fixed and $|a|\leq r(|x|)$.  In general, the state complexity of $N_{n,l}$ is bounded by $O(r(l))\cdot 2^{O(\ell(n,l))}$.

Let us define $\delta'$ by way of describing how the 2qfa $N_{n,|x|}$ operates. Notice that $\delta'$ is a map from $Q'\times \check{\Sigma}\times Q'\times D\times \Xi_{\lambda}$ to $\complex$. Since $M$ is supposed to use at most $\ell(m(x),|x|)$ work tape cells, we use a series $y_0y_1y_2\cdots y_{\ell(m(x),|x|)}$ of $\ell(n,|x|)+1$ tape symbols to express the content $y$ of $M$'s work tape starting with $y_0=\cent$ (left endmarker). Let $B$ indicate a blank of a tape cell. For each index $n\in\nat$, we write $\Sigma_n$ for the set $\{x\in\Sigma^*\mid m(x)=n\}$.
The 2qfa $N_{n,|x|}$ behaves as follows. Let $y=y_0y_1\cdots y'_{\ell(m(x),|x|)}$ and $y'=y'_0y'_1\cdots y'_{\ell(m(x),|x|)}$.

\begin{quote}
On input $x\in\Sigma_n$, $N_{n,|x|}$ starts with the initial inner state $(q_0,0,\cent B^{\ell(n,|x|)},0)$. Inductively, assume that $M$ changes its surface configuration from $(q,j,k,y,t,z)$ to $(p,j',k',y',t',z\xi)$ by moving its input-tape head in direction $d_1$, its work-tape head in direction $d_2$, and its advice-tape head in direction $d_3$, and also by writing $y'_k$ over $y_k$ and changing  $z$ to $z\xi$ on a garbage tape. Corresponding to this step, $N_{n,|x|}$ moves its own input-tape head similarly, changes its inner state from $(q,k,y,t)$ to $(p,k',y',t')$, and modifies $z$ to $z\xi$.
More formally, we define a (quantum) transition function $\delta'$ of $N_{n,|x|}$ by setting $\delta'((q,k,y,t),x_j | (p,k',y',t'),d_1,\xi) = \delta(q,x_j,y_k,a_t | p,y'_k,d_1,d_2,d_3,\xi)$, where $k'=k+d_2$, $t'=t+d_3$, and  $rest_k(y)=rest_k(y')$.
\end{quote}

Notice that the above construction of $N_{n,|x|}$ heavily depends on the fixed advice string $a$.

By the definition of $N_{n,|x|}$,  on input $x$, $N_{n,|x|}$ correctly simulates all steps of $M$ on $(x,a)$ one by one and reaches the same outcome of $M$.
Therefore, for any fixed constant $\varepsilon\in[0,1/2)$, if $M$ accepts (resp., rejects) $(x,h(|x|))$ within $p(n,|x|)$ time using at most $\ell(n,|x|)$ space with error probability at most $\varepsilon$, then $N_{n,|x|}$ also accepts (resp., rejects) $x$ within $p(n,|x|)$ time with error probability at most $\varepsilon$.

Since $N_{n,|x|}$ simulates $M$ with $h$ precisely, if $M$'s configurations quantumly interfere with one another, then their associated $N_{n,|x|}$'s configurations interfere as well. Hence, $N_{n,|x|}$ is indeed the desired  2qfa and the lemma then follows instantly.
\end{proof}

The converse of Lemma \ref{QTM-to-QFA}, Lemma \ref{QFA-to-QTM}, will be shown by giving a simulation of a nonuniform family of 2qfa's by appropriately chosen advised QTMs. In the proof of Lemma \ref{QFA-to-QTM}, nonetheless, in order to make a quantum interference take place correctly, we further need to avoid any time discrepancy caused by the different simulation speed along different computation paths and
to adjust the timing of reaching the same configurations. For this purpose, we need to control the movement of an input-tape head of a 2qfa.
We say that an automaton $M$ \emph{sweeps} a tape or $M$ is a \emph{sweeping} automaton if $M$'s tape head always moves rightward along a circular input tapes in one direction from $\cent$ to $\dollar$, and further to $\cent$.

Another crucial point of the proof of Lemma \ref{QFA-to-QTM} is how to encode all quantum transitions of a 2qfa into a single advice string for the purpose of performing such transitions easily from the information retrieved from this advice string.

\begin{lemma}\label{QFA-to-QTM}
Let $b$, $r$, and $p$ be functions and let $m$ be a size parameter satisfying $m(x)\leq b(|x|)$ for all $x$. Let $\{N_{n,l}\}_{n,l\in\nat}$ denote a nonuniform family of $r(n,l)$-state 2qfa's over a fixed input alphabet $\Sigma$ and a constant-size garbage alphabet.
There exist an advised QTM $M$ and an $O(b(|x|)r(m(x),|x|)^9\log^2p(m(x),|x|))$-bounded advice function $h$ such that, for any input $x$, if $N_{m(x),|x|}$ accepts (resp., rejects) $x$ with bounded-error probability in expected $p(m(x),|x|)$ time, then $M$ accepts (resp., rejects) $(x,h(|x|))$ with bounded-error probability within expected $O(b(|x|)\tilde{r}(m(x),|x|)\tilde{p}(m(x),|x|))$ time using $O(\log{r(m(x),|x|)})$ space, where $\tilde{r}(n,l)= r(n,l)^9\log{r(n,l)}$ and $\tilde{p}(n,l)= p(n,l)\log^2p(n,l)$.
\end{lemma}

\begin{proof}
Let $b$, $r$, $p$, $m$, $\Sigma$, and $\{N_{n,l}\}_{n,l\in\nat}$ be given as in the premise of the lemma.
For each index $n\in\nat$, we set $\Sigma_n =\{x\in\Sigma^*\mid m(x)=n\}$.
Take an arbitrary input $x\in\Sigma^*$ and let $n=m(x)$ for brevity.
Let $N_{n,|x|}$ have the form $(Q_n,\Sigma,\{\cent,\dollar\},\Xi,\delta_n, q_{n,0},Q_{n,qcc},Q_{n,rej})$. We want to define the desired advised QTM $M$,  together with an appropriate advice function $h$, which can  simulate  $N_{n,|x|}$ on all inputs $x$ in $\Sigma_n$
for any index $n\in\nat$.
For our convenience, we identify $Q_n$ with the set $\{0,1\}^{r_1(n,|x|)}$ and $\Xi_{\lambda}$ with $\{0,1\}^{r_2}$ for a certain function $r_1$ and a certain constant $r_2\in\nat^{+}$;  thus, we obtain $r_1(n,|x|) =\log|Q_n|$ ($=\log{r(n,|x|)}$) and $r_2=\log|\Xi_{\lambda}|$, where $\Xi_{\lambda} = \Xi\cup\{\lambda\}$.

The transition function $\delta_n$ of $N_{n,l}$ can be viewed as a \emph{transition table} $T_{n,l}$, as explained in Section \ref{sec:transition-table}, which is a matrix in which each $(q,\sigma)$-row contains a ``description'' of a quantum circuit $C^{(n,l)}_{q,\sigma}$ that takes $\qubit{\phi_0} = \qubit{0^{r_1(n,|x|)}}\qubit{00}\qubit{0^{r_2}}$ as an input and approximates to within $2^{-cp(n,l)}$ a quantum state $\sum_{(p,d,\xi)}  \delta_n(q,\sigma|p,d,\xi) \qubit{p,d,\xi}$.
Let us recall from Section \ref{sec:transition-table} the encoding scheme of such  a transition table into an appropriate binary string.
Using this encoding scheme, we encode $T_{n,l}$ into an advice string $\pair{T_{n,l}}$ of $O(|Q|^9|\Xi_{\lambda}|^8\log^2p(n,l))  \subseteq O(r(n,l)^9\log^2p(n,l))$ symbols, from which we can easily retrieve the necessary information on any transition $\delta_n(q,\sigma|p,d,\xi)=\alpha$.

For each number $l$ in $\nat$, we define the desired advice function $h$ as $h(l) = 1^{1}\#\pair{T_{1,l}} \#^2 1^{2}\#\pair{T_{2,l}} \#^2 \cdots \#^2 1^{\bar{n}_l}\#\pair{T_{\bar{n}_l,l}}$, where $\bar{n}_l = \max\{m(z)\mid z\in\Sigma^*, |z|=l\}$. Since $\bar{n}_l \leq b(l)$, the length of $h(l)$ is bounded by $O(b(l)r(n,l)^9\log^2p(n,l))$.


The following procedure briefly describes the behavior of $M$ with $h$.

\begin{quote}
On input $(x,h(|x|))$, $M$ first computes $n=m(x)$, writes down the string $\cent q_{n,0} \dollar$ on its work tape by sweeping the tape, and tries to simulate $N_{n,|x|}$ on $x$ as follows. Assume that currently $N_{n,|x|}$ is in inner state $q$ scanning the $j$th input tape cell, changes $q$ to $p$, moves its input-tape head in direction $d$, and dumps $\xi$ onto a garbage tape with amplitude $\delta_n(q,x_j | p,d,\xi)$. By sweeping the work tape from $\cent$ to $\dollar$, $M$ searches for the $(q,x_j)$-row, and reads its entry $\pair{C^{(n,|x|)}_{q,x_j}}$ symbol by symbol by performing each quantum gate constituting $C^{(n,|x|)}_{q,x_j}$ to generate a quantum state $\sum_{(p',d',\xi')} \delta(q,x_j | p',d',\xi') \qubit{p',d',\xi'}$.
If an entry $(p,d,\xi)$ appears in this quantum state, then $M$ overwrites  the work tape by $p$, moves the input-tape head in direction $d$, and writes down $\xi$ onto the garbage tape.
To simulate one step of $N_{n,|x|}$, $M$ needs to sweep the advice tape once and sweep the work tape at most $|h(|x|)|$ times. Note that $M$ keeps its input-tape head at a standstill during each sweeping process.
\end{quote}

A quick analysis of the above description shows that  the space usage of $M$'s work tape is bounded by $O(r_1(m(x),|x|)) \subseteq O(\log{r(m(x),|x|)})$ and $M$'s  expected runtime is bounded by  $O(p(m(x),|x|)r_1(m(x),|x|) |h(|x|)|) \subseteq O(b(|x|)r(m(x),|x|)\log{r(m(x),|x|)} p(m(x),|x|) \log^2{p(m(x),|x|)})$.
Moreover, it is not difficult to verify that, if $N_{n,|x|}$ accepts (resp., rejects) $x$ with bounded-error probability, then $M$ accepts (resp., rejects) $(x,h(|x|))$ with bounded-error probability.
\end{proof}

The following lemma relates to the polynomial honesty condition of size parameters. This condition will become quite essential in the proof of Proposition \ref{characterize-PBQL/poly}.

\begin{lemma}\label{poly-honest-bound}
Let $(L,m)$ and $(K,m')$ be two parameterized decision problems and let   $\LL=\{(L_n^{(+)},L_n^{(-)})\}_{n\in\nat}$ be any family of promise decision problems. Let $\CC$ denote an arbitrary nonempty nonuniform state complexity class.
\renewcommand{\labelitemi}{$\circ$}
\begin{enumerate}\vs{-2}
  \setlength{\topsep}{-2mm}%
  \setlength{\itemsep}{1mm}%
  \setlength{\parskip}{0cm}%

\item If $\LL$ is induced from $(L,m)$ and $\LL\in\CC/\poly$, then $m$ is polynomially honest.

\item If $(K,m')$ is induced from $\LL$ and $\LL\in\CC/\poly$, then $m'$ is polynomially honest.
\end{enumerate}
\end{lemma}

\begin{proof}
(1) Assume that $\LL=\{(L_n^{(+)},L_n^{(-)})\}_{n\in\nat}$ is induced from $(L,m)$ and that $\LL$ is in $\CC/\poly$. Since $\LL$ has a polynomial ceiling, there exists a polynomial $r$ such that, for any index $n\in\nat$ and for any input $x$ in $L_n^{(+)}\cup L_n^{(-)}$, $|x|\leq r(n)$ holds. In the case of $x\in L_n^{(+)}\cup L_n^{(-)}$, since $m(x)=n$ by the definition of $(L_n^{(+)},L_n^{(-)})$, we obtain $|x|\leq r(m(x))$. Thus, $m$ is polynomially honest.

(2) Similarly to (1), from $\LL\in\CC/\poly$, we obtain a polynomial $r$ satisfying $|x|\leq r(n)$ for all $n\in\nat$ and all $x\in L_n^{(+)}\cup L_n^{(-)}$. Recall from Section \ref{sec:parameter-promise} that $K_n^{(+)}=\{1^n\# x\mid x\in L_n^{(+)}\}$ and $K_n^{(-)} = \{1^n\# x\mid x\in L_n^{(-)}\}\cup \{z\# x\mid z\in\Sigma^n-\{1^n\},x\in \Sigma_{\#}^*\} \cup \Sigma^n$, where $\Sigma$ is a fixed input alphabet.
We define $s(n)=n+r(n)+1$ for each number $n\in\nat$.
Consider the case where $w$ is of the form $1^n\# x$ for $x\in L_n^{(+)}\cup L_n^{(-)}$.
It then follows that $|w|=|1^n\# x|\leq n+|x|+1\leq n+r(n)+1\leq s(m'(w))$ because of $m'(w)=n$. In any other case, we obtain $|w|= m'(w)$ by the definition of $m'$, and thus $|w|\leq s(m'(w))$ follows.  Therefore, $m'$ is polynomially honest.
\end{proof}

With the help of the supporting lemmas, Lemmas \ref{QTM-to-QFA}--\ref{poly-honest-bound}, we provide the proof of Proposition \ref{characterize-PBQL/poly}(1).

\begin{proofof}{Proposition \ref{characterize-PBQL/poly}(1)}
Hereafter, we intend to present the desired proof only for the case where $\AAA=\mathrm{BQ}$ and $\BB=\bql$.
Let $(L,m)$ be any parameterized decision problem and let $\LL=\{(L_n^{(+)},L_n^{(-)})\}_{n\in\nat}$ be the family of promise decision problems induced from $(L,m)$.
For convenience, we abbreviate $L_n^{(+)}\cup L_n^{(-)}$ as $\Sigma_n$ for each index $n\in\nat$. By the definition of $\LL$, $\Sigma_n$ equals the set  $\{x\in\Sigma^*\mid m(x)=n\}$.

(If--part)
Assuming that $\LL\in \pt\twobq/\poly$, we take polynomials $p$ and $r$, and a family $\NN=\{N_{n}\}_{n\in\nat}$ of 2qfa's having $r(n)$ inner states such that, for every index $n\in\nat$,  $N_n$ solves $(L_n^{(+)},L_n^{(-)})$ with bounded-error probability in expected $p(n,|x|)$ time on all inputs $x$ in $\Sigma_n$. We define a new 2qfa $N'_{n,l}$ simply as $N_n$ restricted to all inputs $x$ of length $l$.
Take a polynomial $b$ satisfying $m(x)\leq b(|x|)$ for all $x$.
Since $m$ is polynomially honest by Lemma \ref{poly-honest-bound}(1), there is a polynomial $q$ for which $|x|\leq q(|x|)$ holds for all $x$.
therefore, $(L,m)$ belongs to $\phsp$.

From Lemma \ref{QFA-to-QTM}, it follows that there exist an  $O(b(|x|)r(m(x))^9\log{p(m(x),|x|)})$-bounded advice function $h$ and an advised QTM $M$ such that, for an arbitrary string $x$,  $M$ on the input $(x,h(|x|))$ simulates $N'_{m(x),|x|}$ on $x$ using space $O(\log{r(m(x))}) \subseteq O(\log{m(x)})$.
In other words, $M(x,h(|x|))$ computes $L_{m(x)}(x)$ with bounded-error probability. Since $L_{m(x)}=L_n$ for all $x\in \Sigma_n$, $M(x,h(|x|))$ actually computes $L(x)$ for all strings $x$ in $\bigcup_{n\in\nat}\Sigma_n$.  Since $\Sigma^* = \bigcup_{n\in\nat}\Sigma_n$, $(L,m)$ belongs to $\para\bql/\poly$.

(Only if--part) Assume that $(L,m)$ is in $\para\pt\bql/\poly \cap \phsp$.
Since $(L,m)\in\phsp$, $m$ is polynomially honest. There exist an advised QTM $M$ and an  $r(|x|)$-bounded advice function $h$ satisfying the following: for  all inputs $x$,  $M(x,h(|x|))$ computes $L(x)$ with bounded-error probability in expected $p(m(x),|x|)$ time using  $O(\log{m(x)})$ work space. Since $m$ is polynomially honest, we take a constant $c>0$ satisfying $|x|\leq m(x)^c+c$ for all $x$. With this value $c$, given any index $n\in\nat$, we write $\tilde{\Sigma}_n$ for the set $\{x\in\Sigma^*\mid |x|\leq n^c+c\}$.

Lemma \ref{QTM-to-QFA} then ensures the existence of a nonuniform family $\NN= \{N_{n,l}\}_{n\in\nat}$ of 2qfa's satisfying the following condition: for any given string $x\in\tilde{\Sigma}_n$, $N_{m(x),|x|}$ on $x$ simulates $M$ on $(x,h(|x|))$ and the state complexity of $N_{m(x),|x|}$ is at most $O(r(|x|))\cdot 2^{O(\log{m(x)})}$, which is upper-bounded by $|x|^{O(1)}\cdot 2^{O(\log{m(x)})} \subseteq m(x)^{O(1)}$
because of $|x|\leq m(x)^c+c$.
For an arbitrary index $n\in\nat$, a new machine $N'_n$ is defined to behave as follows.
Given any input $x$, we first set $l=|x|$ if $x\in\tilde{\Sigma}_n$ and $l=n$ otherwise. This can be done by sweeping an input tape once using  additional $n^{O(1)}$ inner states.
We then run $N_{n,l}$ on $x$.
In particular, when $m(x)=n$, $N'_{n}$ can simulate $N_{n,|x|}$ on every input $x$ in $\tilde{\Sigma}_n$.
Since $\NN$ simulates $M$ with $h$, $N'_{n}$ correctly computes $L(x)$ on all inputs $x$ in $\tilde{\Sigma}_n$ in expected $O(p(n,|x|))$ time.
Since $N'_{n}$ uses $n^{O(1)}$ states, we conclude that $\LL$ falls in $\pt\twobq/\poly$.
\end{proofof}

Next, we wish to prove Proposition \ref{characterize-PBQL/poly}(2), which deals with the parameterized decision problem $(K,m)$ induced from each $\dl$-good  family $\LL=\{(L_n^{(+)},L_n^{(-)})\}_{n\in\nat}$ of promise decision problems.
Notice that, by the definition of $K$, for any $n$ and $x$,  $1^n\# x\in K$ (resp., $1^n\# x\in\overline{K}$) iff $x\in L_n^{(+)}$ (resp., $x\in L_n^{(-)}$).
For the proof of the proposition, we need two more lemmas, Lemmas \ref{para-promise-QTM}--\ref{para-promise-2qfa}, which look analogous to Lemmas \ref{QTM-to-QFA}--\ref{QFA-to-QTM}.

\begin{lemma}\label{para-promise-QTM}
Let $M$ be an advised QTM, let $r$ and $\ell$ be two  functions, and let $h$ be an advice function with $|h(n)|\leq r(n)$ for all $n\in\nat$. There exists a nonuniform family $\{N_{n,l}\}_{n,l\in\nat}$ of $O(r(l))\cdot 2^{O(\ell(n,l))}$-state 2qfa's such that, for any $n$ and $x$, if $M$ accepts (resp., rejects) $(1^{n}\# x,h(|x|))$ with bounded-error probability in expected  $p(n,|x|)$ time using at most $\ell(n,|x|)$ space, then $N_{n,|x|}$ accepts (resp., rejects) $x$ with bounded-error probability in expected $O(p(n,|x|))$ time.
\end{lemma}

\begin{proof}
Let $r$, $h$, $p$, and $M$ be given as in the premise of the lemma. Let $\Sigma$ denote an input alphabet used by $M$.
Slightly different from the proof of Lemma  \ref{QTM-to-QFA}, we construct the desired family $\{N_{n,l}\}_{n,l\in\nat}$ of 2qfa's to work as follows.
Given an index $n\in\nat$ and an input string $x\in\Sigma^*$, $N_{n,|x|}$  generates both $1^n\# x$ and $h(|x|)$ and then runs $M$ on $(1^n\# x,h(|x|))$ to produce an outcome. This behavior of $N_{n,|x|}$ is possible because $n$ and $|x|$ are fixed for $N_{n,|x|}$ and we can store $h(|x|)$ in the form of inner states.
Clearly, $N_{n,|x|}$ correctly outputs $M(1^{n}\# x, h(|x|))$ on all inputs $x$ in expected $O(p(n,|x|))$ time since $M$  halts on $(1^{n}\# x,h(|x|))$ in expected $p(n,|x|)$ time.
\end{proof}

\begin{lemma}\label{para-promise-2qfa}
Let $r$ be any function and let $\{N_{n,l}\}_{n,l\in\nat}$ denote any  nonuniform family of 2qfa's with $r(n,l)$ inner states. There is an advised QTM $M$ and an $O(nr(n,|x|)^9\log^2{p(n,|x|)})$-bounded advice function $h$ such that, for any $x$, if $N_{n,|x|}$ accepts (resp., rejects) $x$ with bounded-error probability in expected  $p(n,|x|)$ time, then $M$ accepts (resp., rejects) $(1^{n}\# x,h(|x|))$ with bounded-error probability in expected $O(nr(n,|x|)^9 \tilde{p}(n,|x|))$ time using $O(\log{r(n,|x|)})$ space, where $\tilde{p}(n,l) = p(n,l)\log^2p(n,l)$.
\end{lemma}

\begin{proof}
Let $r$ and $\{N_{n,l}\}_{n,l\in\nat}$ be given as in the premise of the lemma. With a construction similar to the proof of Lemma  \ref{QFA-to-QTM},  we define the desired advised QTM $M$ that behaves as  follows in the presence of an appropriate advice function $h$. On input $w$, we first check whether $w$ is of the form $1^n\# x$ for certain $n$ and $x$. If not, then we reject $w$ instantly. Assuming that $w = 1^n\# x$, we compute $|x|$, retrieve the description of $N_{n,|x|}$ from $h(|x|)$ and run $N_{n,|x|}$ on $x$. It is not difficult to see that $M$ accepts (resp., rejects) $(1^n\# x,h(|x|))$ with bounded-error probability iff $N_{n,|x|}$ accepts (resp., rejects) $x$ with bounded-error probability. The expected runtime and the space usage of $M$ can be estimated similarly to the proof of  Lemma  \ref{QFA-to-QTM}.
\end{proof}

The proof of Proposition \ref{characterize-PBQL/poly}(2) is in essence similar to the proof of Proposition \ref{characterize-PBQL/poly}(1) except for the use of $(K,m)$ that is induced from a given family  $\LL$ of promise decision problems.

\begin{proofof}{Proposition \ref{characterize-PBQL/poly}(2)}
Let $\LL=\{(L_n^{(+)},L_n^{(-)})\}_{n\in\nat}$ be any $\dl$-good family of promise decision problems and let $(K,m)$ denote the parameterized decision problem induced from $\LL$. For convenience, we write $\Theta$ for the set $\{1^n\# x\mid n\in\nat, x\in L_n^{(+)}\cup L_n^{(-)}\}$. Since $\Theta$ is in $\dl$, by the definition of $(K,m)$, $m$ is a log-space size parameter, and thus it is  polynomially bounded.
The following proof is meant for the case of $\AAA=\mathrm{BQ}$
and $\BB=\bql$ in accordance with the proof of Proposition \ref{characterize-PBQL/poly}(1).

(If--part) Assume that $\LL$ is in $\pt\twobq/\poly$. There exist two  polynomials $p$ and $r$ as well as a nonuniform family $\NN=\{N_n\}_{n\in\nat}$ of $O(r(n))$-state 2qfa's solving $\LL$ in expected $O(p(n,|x|))$ time with bounded-error probability.
By Lemma \ref{poly-honest-bound}(2), $m$ is polynomially honest, and thus the problem $(K,m)$ belongs to $\phsp$.

For the 2qfa family $\NN$, Lemma \ref{para-promise-2qfa} guarantees the existence of an advised QTM $M$ and an $O(nr(n)^9\log^2{p(n,|x|)})$-bounded advice function $h$ such that $M$ on input $(1^{n}\# x,h(|x|))$ simulates $N_{n,|x|}$ on $x$ in  expected $O(nr(n,|x|)^9\log^2{p(n,|x|)})$ time using $O(\log{r(n,|x|)})$ space.
We further modify $M$ so that it first checks whether an input $w$ of the form $1^{n}\# x$ belongs to $\Theta$ using logarithmic space; if not, it immediately rejects the input. This modified machine makes bounded-error probability on all inputs $w$ and thus recognizes $K$ with only additional $O(\log{|w|})$ space.
Since $O(\log{r(n,|x|)}+\log|w|)\subseteq O(\log{m(w)})$, $(K,m)$ belongs to $\para\pt\bql/\poly$.

(Only if--part)
Assuming that $(K,m)\in\para\pt\bql/\poly\cap \phsp$, we take an advised QTM $M$ and an $r(|w|)$-bounded advice function $h$ for which $M$ solves $(K,m)$ using $h$ with bounded-error probability in expected $p(m(w),|w|)$ time using at most $\ell(m(w))$ space for any input $w$, where $p$ and $\ell$ are respectively a polynomial and a logarithmic function. Additionally, we assume that $\ell$ is nondecreasing.
The polynomial honesty of $m$ comes from the assumption of $(K,m)\in\phsp$. Let us take a polynomial $q$ satisfying $|w|\leq q(m(w))$ for any $w$.
Note that, for any $x\in L_n^{(+)}\cup L_n^{(-)}$, $|x|\leq |1^n\# x|\leq q(m(1^n\# x))= q(n)$ since $m(1^n\# x)=n$.
Moreover, if $w=1^n\# x$, then it follows that
$p(m(w),|w|)\leq p(n,q(n))$ and $\ell(m(w),|w|)\leq \ell(n,q(n))$.

Lemma \ref{para-promise-QTM} then provides a nonuniform family $\NN=\{N_{n,l}\}_{n,l\in\nat}$ of 2qfa's having $O(r(|x|))\cdot 2^{O(\ell(n,q(n)))}$ inner states for which each $N_{n,|x|}$ on input $x$  simulates $M$ on $(1^{n}\# x,h(|x|))$ correctly
in expected $O(p(n,q(n)))$ time.
Thus, $\NN$ solves $\LL$ with bounded-error probability.
We note that $O(p(n,q(n)))\subseteq n^{O(1)}$ and $2^{O(\ell(n,q(n)))} \subseteq n^{O(1)}$. This immediately concludes that $\LL$ belongs to $\pt\twobq/\poly$.
\end{proofof}

\subsection{Proof of Theorem \ref{general-theorem}}\label{sec:proof-of-theorem}

Finally, we are ready to describe the proof of Theorem \ref{general-theorem}.
We have already proven a key claim, Proposition \ref{characterize-PBQL/poly}, in Section \ref{sec:role-of-advice}. For the intended proof of the theorem, however, we still need two more supporting claims regarding polynomial-size advice.

\begin{lemma}\label{NL-vs-NL/poly}
Let $(\AAA,\BB)\in\{(\nl,\dl), (\nl,\bpl), (\nl,\bql), (\bql,\bpl)\}$. It then follows that $\AAA/\poly \subseteq \BB/\poly$ iff $\AAA \subseteq \BB/\poly$. The same holds even if the expected runtime of underlying Turing machines is limited to polynomials.
\end{lemma}

\begin{proof}
In what follows, we intend to show the lemma only for
the case of $\AAA=\nl$ and $\BB=\bql$ since the other cases are similarly proven.

(Only If--part) The implication of $\nl/\poly\subseteq \bql/\poly$ to $\nl\subseteq \bql/\poly$ is obviously true because $\nl$ is included in $\nl/\poly$.

(If--part) Assume that $\nl\subseteq \bql/\poly$. Our goal is to verify that $\nl/\poly\subseteq \bql/\poly$. Let us focus on an arbitrary language $L$ in $\nl/\poly$ and take a polynomial $p$ and a log-space NTM $M$ satisfying the following two conditions for every input $x$: (i) $|h(|x|)|\leq p(|x|)$ and (ii) $N$ accepts $(x,h(|x|))$ iff $x$ is in $L$.
We define a new language $K$ as $K=\{(x,s)\mid |s|\leq p(|x|), \text{ $N$ accepts $(x,s)$}\}$. Since  $K$ clearly belongs to $\nl$, our assumption implies that $K \in \bql/\poly$.
We then take a log-space advised QTM $M$ that recognizes $K$ with bounded-error probability using an appropriate polynomially-bounded advice function $g$; that is, $M(x,s,g(|x|,|s|))$ computes $K(x,s)$ with bounded-error probability for any pair $(x,s)$.
Let us define a new advice function $r$ by setting $r(n)= \pair{h(n),g(n,|h(n)|)}$ and design a new QTM $\tilde{N}$ so that it starts with $(x,r(|x|))$ and simulates $M$ on $(x,h(|x|),g(|x|,|h(|x|)|))$.
Clearly, $\tilde{N}$ with $r$ recognizes $L$
with bounded-error probability.
Take a polynomial $k$ satisfying $|g(n,n')|\leq k(n,n')$ for any pair $n,n'\in\nat$.
Note that $|r(n)|$ is $O(|h(n)|+|g(n,|h(n)|)|)$, which is included in $O(p(n)k(n,p(n)))$. Since $\tilde{N}$ is a bounded-error advised QTM for $L$ with the polynomially-bounded advice function $r$, $L$  belongs to $\bql/\poly$. Therefore, we conclude that $\nl/\poly \subseteq \bql/\poly$, as requested.
\end{proof}

Another claim stated below connects between parameterized complexity classes and standard complexity classes under the presence of advice.

\begin{lemma}\label{para-NL-vs-NL}
Let $(\AAA,\BB)\in\{(\nl,\dl), (\nl,\bpl), (\nl,\bql), (\bql,\bpl)\}$. It then follows that $\para\AAA/\poly \cap \phsp \subseteq \para\BB/\poly$ iff $\AAA/\poly \subseteq \BB/\poly$. The same holds even if the expected runtime of underlying Turing machines is limited to polynomials.
\end{lemma}

Recall that, by our convention stated in Section \ref{sec:main-contribution},  $\para\pt\nl/\poly$ and $\pt\twon/\poly$ mean $\para\nl/\poly$ and $\twon/\poly$, respectively.

\begin{proofof}{Lemma \ref{para-NL-vs-NL}}
Hereafter, we are focused only on the case where $\AAA=\nl$ and $\BB=\bql$ in accordance with the proof of Lemma \ref{NL-vs-NL/poly}.

(Only If--part) We begin with assuming that $\para\nl/\poly \cap \phsp \subseteq \para\bql/\poly$.
Let us take the binary size parameter $m_{bin}(x)=|x|$ defined for all strings  $x$ and consider any parameterized decision problem $(L,m_{bin})$ satisfying $L\in\nl/\poly$.
Notice that $m_{bin}$ is polynomially honest.
Since $(L,m_{bin})\in\para\nl/\poly\cap \phsp$, our assumption concludes that $(L,m_{bin})\in \para\bql/\poly$. This is logically equivalent to $L\in\bql/\poly$ by the definition of $m_{bin}$.
Since $L$ is arbitrary, we conclude that
$\nl/\poly \subseteq \bql/\poly$.

(If--part) On the contrary, assume that $\nl/\poly\subseteq \bql/\poly$. Let $(L,m)$ be any parameterized decision problem in $\para\nl/\poly\cap \phsp$.
By the definition, $m$ is a log-space size parameter, and thus
we can deterministically compute $1^{m(x)}$ from $x$ using $O(\log|x|)$ space.
Since $m$ is polynomially bounded, an appropriately chosen polynomial $q$ can satisfy $m(x)\leq q(|x|)$ for all strings $x$. The polynomial honesty of $m$ also guarantees the existence of a constant $c\geq0$ satisfying $|x|\leq m(x)^c+c$ for all $x$.
Since $(L,m)\in\para\nl/\poly$, there exists an NTM $N$ and a polynomially-bounded advice function $k$ such that $N$ with $k$ solves $L$ using $O(\ell(m(x)))$ space for a certain logarithmic function $\ell$.
Without loss of generality, we can assume that $\ell$ is nondecreasing.

Let us define $K'=\{(x,1^t)\mid x\in L, t\in\nat, m(x)\leq t\}$.
We wish to claim that $K'$ belongs to $\nl/\poly$.
To verify this claim, let us consider a new NTM that behaves as follows.
On input $w$ of the form $(x,1^t)$, firstly check whether $m(x)\leq t$. If so, then output the value $L(x)$; otherwise, immediately reject $x$.
With the help of the advice string $k(|x|)$, we can check that ``$x\in L$?'' using space $O(\ell(m(x))) \subseteq O(\ell(q(|w|)))$
and also check that ``$m(x)\leq t$?'' using space $O(\log|x|+\log{t})\subseteq O(\log{|w|})$.
Hence, this NTM requires only $O(\log{|w|})$ space. From this result, we conclude that $K'$ belongs to $\nl/\poly$.

Since $K'\in\nl/\poly$, our assumption further implies that $K'\in\bql/\poly$.
This provides us with a logarithmic function $\ell'$, a polynomially-bounded advice function $h$,  and an advised QTM $M$ that recognizes $\{(x,h(|x|))\mid x\in K'\}$ with bounded-error probability using at most $\ell'(|x|)$ space.
Finally, we design a new machine $N$ that behaves as follows: on input $x$, compute $n=m(x)$ and run $M$ with $h$ on $(x,1^n)$. Note that $N$ is also an advised QTM and runs using space $O(\ell'(|x|)+\log|x|) \subseteq O(\log{m(x)})$ since $|x|\leq m(x)^c+c$.
It thus follows that  $(L,m)$ is in $\para\bql/\poly$.
\end{proofof}

Let us present the proof of Theorem \ref{general-theorem}, which is now a relatively easy consequence of Lemmas \ref{NL-vs-NL/poly}--\ref{para-NL-vs-NL} and Proposition \ref{characterize-PBQL/poly}.

\vs{-2}
\begin{proofof}{Theorem \ref{general-theorem}}
In what follows, we intend to show only the case of $\AAA=\mathrm{N}$ and $\BB=\mathrm{BQ}^{\dagger}$ because the other cases can be proven in a similar manner.

(If--part) Firstly, we assume that $\nl\subseteq \pt\bql/\poly$.
Since this assumption is, by Lemma \ref{NL-vs-NL/poly}, logically equivalent to $\nl/\poly\subseteq \pt\bql/\poly$,
Lemma \ref{para-NL-vs-NL} implies that $\para\nl/\poly\cap\phsp \subseteq \para\pt\bql/\poly$. Using this last inclusion, it suffices for us to verify that $\twon/\poly \subseteq \pt\twobq^{\dagger}$.

For our purpose, let us consider an arbitrary family $\LL= \{(L_n^{(+)},L_n^{(-)})\}_{n\in\nat}$ of promise decision problems in $\twon/\poly$. Since $\twon/\poly$ is $\dl$-good, we can take an $\dl$-good extension $\hat{\LL}=\{(\hat{L}_n^{(+)},\hat{L}_n^{(-)})\}_{n\in\nat}$ of $\LL$ in $\twon/\poly$.
We write $\Sigma_n$ for $\hat{L}_n^{(+)}\cup \hat{L}_n^{(-)}$ for each index $n\in\nat$.
There exists a nonuniform family $\{M_n\}_{n\in\nat}$ of 2nfa's with polynomially many inner states such that, for any index $n\in\nat$, $M_n$ solves $(\hat{L}_n^{(+)},\hat{L}_n^{(-)})$ on all inputs in $\Sigma_n$ with bounded-error probability. Consider the parameterized decision problem $(K,m)$ induced from $\hat{\LL}$. Since $\hat{\LL}\in\twon/\poly$, by Lemma \ref{poly-honest-bound}(2), $m$ is polynomially honest.
We apply Proposition \ref{characterize-PBQL/poly}(2) and then obtain $(K,m)\in \para\nl/\poly \cap \phsp$, from which our assumption further places $(K,m)$ in $\para\pt\bql/\poly$. Proposition \ref{characterize-PBQL/poly}(2) then concludes that $\hat{\LL}$ belongs to $\pt\twobq/\poly$.
Since $\hat{\LL}$ is an $\dl$-good extension of $\LL$, $\LL$ must belong to $\pt\twobq^{\dagger}/\poly$.
Because $\pt\twobq^{\dagger}/\poly \subseteq \pt\twobq^{\dagger}$, the desired inclusion $\twon/\poly \subseteq \pt\twobq^{\dagger}$ follows immediately.

(Only If--part) On the contrary, we assume that $\twon/\poly \subseteq \pt\twobq^{\dagger}$. This implies that $\twon/\poly\subseteq \pt\twobq^{\dagger}/\poly$.
It therefore suffices to prove that $\para\nl/\poly\cap\phsp \subseteq \para\pt\bql/\poly$ because, once this is proven, Lemma \ref{para-NL-vs-NL} implies that $\nl/\poly \subseteq \pt\bql/\poly$ and Lemma \ref{NL-vs-NL/poly} further concludes that $\nl\subseteq \pt\bql/\poly$.

Let us take any parameterized decision problem $(L,m)$ in $\para\nl/\poly\cap \phsp$. Let $\LL$ denote the family of promise decision problems induced from $(L,m)$.
Proposition \ref{characterize-PBQL/poly}(1) implies that $\LL\in\twon/\poly$.
Our assumption then leads to the conclusion that $\LL\in\pt\twobq^{\dagger}/\poly$. It then follows by Proposition \ref{characterize-PBQL/poly}(1) that $(L,m)$ belongs to $\para\pt\bql/\poly$.
\end{proofof}

From Theorem \ref{general-theorem}, Corollary \ref{NL-equivalence} follows immediately. This theorem  also leads to the main result of
\cite[Theorem 1.1]{Kap14} (shown here as Corollary \ref{N-D-case}) whose proof relies on the property of a particular $\nl$-complete problem, called the two-way liveness problem \cite{SS78}. Our key argument in the proof of
Proposition \ref{characterize-PBQL/poly}, however, uses the parameterized complexity classes $\para\dl$ and $\para\nl$ as well as their properties, which are not directly relying on any particular $\nl$-complete problem.

\begin{corollary}\emph{\cite{Kap14}}\label{N-D-case}
$\twon/\poly \subseteq \twod$ iff $\nl\subseteq \dl/\poly$.
\end{corollary}

\section{Quantum Advice and Quantum Transition Tables}\label{sec:quantum-advice}

Since Theorem \ref{two-way-equivalence} concerns \emph{quantum advice}, for the  proof of the theorem, we first examine the basic properties of quantum advice. It was shown in \cite[Lemma 3.1]{NY04} that a polynomial-time QTM with quantum advice can be translated into an ``equivalent'' polynomial-size quantum circuit family starting with additional quantum states.
In the case of quantum finite automata, nevertheless, we rather intend to quantize transition tables into ``superpositions'' of such transition tables
and feed them to quantum finite automata  so that different transition tables  may regulate different behaviors of quantum finite automata.
Firstly, we need to clarify the notion of a \emph{superposition of transition tables} or, more succinctly, a \emph{quantum transition table}.
As discussed in Section \ref{sec:transition-table}, a transition table $T$ of a 2qfa can be encoded into a binary string $\pair{T}$, from which an appropriately designed 2qfa can generate from $(q,\sigma)$ a quantum circuit $C^{(n,l)}_{q,\sigma}$ representing a target transition  and approximately execute this transition by simply running $C^{(n,l)}_{q,\sigma}$ on the quantum state $\qubit{q,\sigma}$.
To handle various transition tables $T_1,T_2,\ldots,T_k$ at once, we can create a superposition $\qubit{\phi}$ of the form $\sum_{i=1}^{k}\alpha_i\qubit{T_i}$ with the use of appropriately chosen amplitudes $\{\alpha_i\}_{i\in[k]}$ satisfying $\sum_{i\in[k]} |\alpha_i|^2=1$.
For technical reason, we further ``encode'' such a superposition $\qubit{\phi}$ of transition tables
into another superposition $\qubit{\phi^{(code)}} = \sum_{i\in[k]}\alpha_i\qubit{\pair{T_i}}$. For convenience, we define the \emph{encoded length} of  $\qubit{\phi}$ to be $\max_{i\in[k]} |\pair{T_i}|$.

To execute a quantum transition table,
we also need to expand our underlying 2qfa's to ``super'' 2qfa's.
A \emph{2-way super quantum finite automaton with a flexible garbage tape} (abbreviated as a 2sqfa) is an octuple $(Q,\Sigma,\{\cent,\dollar\},\Xi, \qubit{\phi},q_0,Q_{acc},Q_{rej})$, where $\qubit{\phi}$ is a quantum transition table and the rest is similar to a 2qfa.
A 2sqfa takes an input and, starting with the inner state $\qubit{q_0}$, makes  transitions described by the quantum transition table $\qubit{\phi}$.
Similar to 2qfa's, at every step, the 2sqfa must observe its inner state to check whether or not it is in a halting state.
We further consider a nonuniform family $\{M_n\}_{n\in\nat}$ of such 2sqfa's.

The notation $\pt\twosbq$ denotes the collection of all nonuniform families of promise decision problems, each family of which is solved  with  bounded-error probability in expected polynomial time by a certain nonuniform family of 2sqfa's having polynomially many inner states. By restricting our interest only on inputs having polynomial ceilings, we obtain $\pt\twosbq/\poly$ from $\pt\twosbq$ in a way similar to getting
$\pt\twobq/\poly$ from $\pt\twobq$.

Our purpose is to give the proof of Theorem \ref{two-way-equivalence}; for this purpose, we wish to prepare two key statements, Proposition \ref{character-2sPBQ} and Lemma \ref{parameter-PBQL-eliminate}.

\begin{proposition}\label{character-2sPBQ}
\renewcommand{\labelitemi}{$\circ$}
Let $(L,m)$ and $(K,m')$ be two parameterized decision problems and let $\LL=\{(L_n^{(+)},L_n^{(-)})\}_{n\in\nat}$ denote a nonuniform family of promise decision problems.
\begin{enumerate}\vs{-1}
  \setlength{\topsep}{-2mm}%
  \setlength{\itemsep}{1mm}%
  \setlength{\parskip}{0cm}%

\item Assume that $\LL$ is induced from $(L,m)$. It then follows that  $(L,m)\in\para\pt\bql/\qpoly\cap \phsp$ iff $\LL\in\pt\twosbq/\poly$.

\item Assume that $\LL$ is $\dl$-good and $(K,m')$ is induced from $\LL$. It then follows that $(K,m')\in\para\pt\bql/\qpoly\cap \phsp$ iff  $\LL\in\pt\twosbq/\poly$.
\end{enumerate}
\end{proposition}

For readability, we postpone the proof of Proposition \ref{character-2sPBQ} until the end of this section. Meanwhile, we demonstrate another useful lemma concerning classical and quantum advice.

\begin{lemma}\label{parameter-PBQL-eliminate}
Given any $\AAA\in\{\bql,\pt\bql\}$, it follows that $\para\AAA/\poly \cap \phsp = \para\AAA/\qpoly\cap\phsp$ iff $\AAA/\poly = \AAA/\qpoly$.
\end{lemma}

\begin{proof}
A key idea of the following proof is similar to that of Lemma \ref{para-NL-vs-NL}. In what follows, we target the case of $\AAA=\bql$ and skip the other case of $\AAA = \pt\bql$.

(Only if--part) Assume that $\para\bql/\poly\cap\phsp$ equals $\para\bql/\qpoly\cap \phsp$. Since the inclusion $\bql/\poly \subseteq \bql/\qpoly$ obviously holds, it suffices  to verify that $\bql/\qpoly \subseteq \bql/\poly$. Let $L$ be any language in $\bql/\qpoly$ over
an appropriate alphabet $\Sigma$.
By setting $m_{bin}(x)=|x|$ for all $x\in\Sigma^*$, we instantly obtain the membership of $(L,m_{bin})$ to $\para\bql/\qpoly \cap\phsp$. By our assumption, $(L,m_{bin})$ falls into $\para\bql/\poly$.
By the definition of $m_{bin}$, $L$ must belong to $\bql/\poly$.

(If--part) We start with assuming $\bql/\poly = \bql/\qpoly$. Let us consider any parameterized decision problem $(L,m)$ in $\para\bql/\qpoly\cap\phsp$.
We remark that $m$ is polynomially honest.
Take a polynomially-bounded quantum advice function $h$ and an advised QTM $M$ that solves $L$ on inputs of the form $(x,h(|x|))$ with bounded-error probability using $O(\log{m(x)})$ space.
Our goal is to demonstrate that $(L,m)$ belongs to $\para\bql/\poly$.
Notice that, since $m$ is a log-space size parameter, $m$ must be polynomially bounded.
This fact ensures the existence of a constant $c>0$ that forces $m(x)\leq|x|^c+c$ to hold for all $x\in\Sigma^*$. Thus, the space usage of $M$ is at most   $O(\log{m(x)})\subseteq O(\log|x|)$  on all inputs of the form $(x,h(|x|))$. Since $h$ is also polynomially bounded, we conclude that $L$ belongs to $\bql/\qpoly$.

Our assumption then yields $L\in\bql/\poly$. Take an advised QTM $N$ and a polynomially-bounded classical-advice function $k$ such that $N$ solves $L$ on all inputs of the
form $(x,k(|x|))$ with bounded-error probability using $O(\log|x|)$ space.
By the polynomial honesty of $m$, there exists another constant $e>0$ satisfying $|x|\leq m(x)^e+e$ for all $x$.
It then follows that $O(\log|x|) \subseteq O(\log(m(x)^e+e)) \subseteq O(\log{m(x)})$. Thus, the space usage of $N$ is upper-bounded by $O(\log{m(x)})$. We therefore conclude that $(L,m)$ is actually in $\para\bql/\poly\cap\phsp$.
\end{proof}

Theorem \ref{two-way-equivalence} follows from Propositions \ref{characterize-PBQL/poly} and \ref{character-2sPBQ} by an additional application of Lemma  \ref{parameter-PBQL-eliminate}.

\begin{proofof}{Theorem \ref{two-way-equivalence}}
(If--part) Assume that $\pt\bql/\poly = \pt\bql/\qpoly$.
Lemma \ref{parameter-PBQL-eliminate} then implies that  $\para\pt\bql/\poly\cap\phsp$ coincides with  $\para\pt\bql/\qpoly\cap\phsp$. Let $\LL$ denote any family of promise decision problems in $\pt\twosbq^{\dagger}/\poly$ and take its $\dl$-good extension $\hat{\LL}$ in $\pt\twosbq^{\dagger}/\poly$.
Let $(K,m)$ denote the parameterized decision problem induced from $\hat{\LL}$. By Proposition \ref{character-2sPBQ}(2), it follows that  $(K,m)$ belongs to $\para\pt\bql/\qpoly \cap\phsp$. By our assumption,  $(K,m)$ falls in $\para\pt\bql/\poly$. Proposition \ref{characterize-PBQL/poly}(2) then implies that $\hat{\LL}\in\pt\twobq/\poly$. By the property of $\dl$-good extension, we obtain $\LL\in\pt\twobq^{\dagger}/\poly$.
Therefore, we conclude that $\pt\twosbq^{\dagger}/\poly \subseteq \pt\twobq^{\dagger}$.

(Only If--part)
Assume that $\pt\twosbq^{\dagger}/\poly\subseteq \pt\twobq^{\dagger}$. Owing to Lemma \ref{parameter-PBQL-eliminate}, it suffices to verify that $\para\pt\bql/\poly\cap\phsp$ equals $\para\pt\bql/\qpoly\cap\phsp$.
It is obvious that $\para\pt\bql/\poly \subseteq \para\pt\bql/\qpoly$.
To show the converse inclusion, let $(L,m)$ denote any parameterized decision problem in $\pt\para\bql/\qpoly\cap \phsp$.
Additionally, take the family $\LL$ of promise decision problems induced from $(L,m)$. By Proposition \ref{character-2sPBQ}(1),  $\LL$ falls in $\pt\twosbq/\poly$.
Since $\LL$ is $\dl$-good, it belongs to $\pt\twosbq^{\dagger}/\poly$. Our assumption further places $\LL$ in $\pt\twobq^{\dagger}/\poly$. Proposition \ref{characterize-PBQL/poly}(1) then  concludes that  $(L,m)\in\para\pt\bql/\poly\cap\phsp$, as requested.
\end{proofof}

The proof of Proposition \ref{character-2sPBQ}(1) requires one more crucial lemma, Lemma \ref{QTM-quantum-advice}, which is analogous to Lemmas \ref{QTM-to-QFA}--\ref{QFA-to-QTM}; however, due to the use of 2sqfa's, we need to deal with ``transition tables'' in Lemma \ref{QTM-quantum-advice} instead of ``transition functions'' used in Lemmas \ref{QTM-to-QFA}--\ref{QFA-to-QTM}.

\begin{lemma}\label{QTM-quantum-advice}
Assume that $m$ is a size parameter, which is polynomially honest. Let $b$, $r$, $s$,  $t$, and $\ell$ denote arbitrary functions.
\renewcommand{\labelitemi}{$\circ$}
\begin{enumerate}\vs{-2}
  \setlength{\topsep}{-2mm}%
  \setlength{\itemsep}{1mm}%
  \setlength{\parskip}{0cm}%

\item Let $M$ be an advised QTM $M$ and, for each index $l\in\nat$, define $h(l)$ to be a superposition $\qubit{\phi_{l}}$ of advice strings in $\Theta^{r(l)}$, where $\Theta$ denotes a fixed advice alphabet. There exist a nonuniform family $\NN= \{N_{n,l}\}_{n,l\in\nat}$ of 2sqfa's with $O(r(l))\cdot 2^{O(\ell(n,l))}$ inner states and a family $\{\qubit{\psi_{n,l}}\}_{n,l\in\nat}$ of superpositions $\qubit{\psi_{n,l}}$ of $\NN$'s transition tables of encoded length $O(r(n,l)^9 \log^2t(n,l))\cdot  2^{O(\ell(n,l))}$  satisfying the following: for any string $x$, if $M$ accepts (resp., rejects) $(x,h(|x|))$ with bounded-error probability in expected $t(m(x),|x|)$ time using  at most $\ell(m(x),|x|)$ space, then $N_{m(x),|x|}$  accepts (resp., rejects) $x$ following the quantum transition table $\qubit{\psi_{m(x),|x|}}$ with bounded-error probability in expected  $O(t(m(x),|x|))$ time.

\item Let $\NN= \{N_{n,l}\}_{n,l\in\nat}$ be a nonuniform family of $r(n,l)$-state 2sqfa's and let $\{\qubit{\psi_{n,l}}\}_{n,l\in\nat}$ be a family of superpositions $\qubit{\psi_{n,l}}$ of $\NN$'s transition tables of encoded length at most $s(n,l)$. Assume that $m(x)\leq b(|x|)$ for all $x$. There exist an advised QTM $M$ and an $O(b(|x|)s(m(x),|x|))$-bounded quantum advice function $h$ such that, for any string $x$, if $N_{m(x),|x|}$  accepts (resp., rejects) $x$ following $\qubit{\psi_{m(x),|x|}}$ with bounded-error probability in expected $t(m(x),|x|)$ time, then $M$ accepts (resp., rejects) $(x,1^{m(x)},h(|x|))$ with bounded-error probability in expected $O(b(|x|)s(m(x),|x|) t(m(x),|x|) \log{r(m(x),|x|)})$ time using  $O(\log{r(m(x),|x|)})$ space.
\end{enumerate}
\end{lemma}

\begin{proof}
(1) Let $m$, $r$, $t$, $\ell$, $M$, $h$, and $\Theta$ satisfy the premise of the lemma.
Assume that $M$ is of the form $(Q,\Sigma,\{\cent,\dollar\}, \Gamma, \Theta,\Xi,\delta,q_0, Q_{acc},Q_{rej})$.
Similarly to the proof of Lemma \ref{QTM-to-QFA}, we want to construct from $M$ a nonuniform family $\NN= \{N'_{n,l}\}_{n,l\in\nat}$ of 2sqfa's.
For each index $l\in\nat$, we assume that $h(l)$ has the form $\sum_{a:|a|=r(l)} \alpha_a\qubit{a}$ with $\sum_{a:|a|=r(l)}|\alpha_a|^2=1$, where $a$ ranges over $\Theta^{r(l)}$.

Let $a$ denote any advice string in $\Theta^{r(l)}$. We write $M_a$ to denote the QTM $M$ whose advice tape consists of $a$.
We first construct a ``2qfa''  $N_{n,|x|,a}$ from $M_a$ in a way similar to the proof of Lemma \ref{QTM-to-QFA} except that we translate its transition function $\delta$ into a transition table $T_{n,l,a}$ partly composed of the descriptions of quantum circuits, as discussed in Section \ref{sec:transition-table}.
This new machine $N_{n,l,a}$ has $O(r(l))\cdot 2^{O(\ell(n,l))}$
inner states and
$T_{n,l,a}$ can be expressed as its encoded string $\pair{T_{n,l,a}}$ of length $O(|Q|^9|\Xi|^8\log^2{t(n,l)})$, which is further bounded by $O(r(l)^9  \log^2t(n,l))\cdot 2^{O(\ell(n,l))}$.
From such transition tables, we define
$\qubit{\psi_{n,l}}$ to be the superposition $\sum_{a:|a|=r(l)}\alpha_a\qubit{T_{n,l,a}}$, which is used as a quantum transition table for $N_{n,l,a}$.

Finally, we define a 2sqfa $N'_{n,|x|}$ as follows. With the use of the  quantum transition table $\qubit{\psi_{n,|x|}}$, we start with an input $\qubit{x}$ and, for each transition table $T_{n,|x|,a}$ in  $\qubit{\psi_{n,|x|}}$,
we run $N_{n,|x|,a}$ on $x$ by following the transition instructions of  $T_{n,|x|,a}$.
This new machine $N'_{n,|x|}$ is the desired 2sqfa.

(2) From the premise of the lemma, we take $\{N_{n,l}\}_{n,l\in\nat}$ and $\{\qubit{\psi_{n,l}}\}_{n,l\in\nat}$.
For brevity, let $e=\ceilings{\log|Q||\Xi_{\lambda}|}+2$.
Assume that each superposition $\qubit{\psi_{n,l}}$ has the form $\sum_{a:|a|=r(n,l)}\alpha_a\qubit{T_{n,l,a}}$, where $T_{n,l,a}$ is $N_{n,l}$'s transition table of encoded length at most $s(n,l)$.
The encoding of $\qubit{\psi_{n,l}}$, denoted by $\qubit{\psi^{(code)}_{n,l}}$, is the superposition $\sum_{a:|a|=r(n,l)} \alpha_a\qubit{\pair{T_{n,l,a}}}$, where $\pair{T_{n,l,a}}$ is the encoded string of $T_{n,l,a}$.

Our goal is to construct the desired advised QTM $M$ and the desired  quantum advice function $h$. The function $h$ is simply defined by setting $h(l)$ to be $\qubit{1^1}\qubit{\#} \qubit{\psi^{(code)}_{1,l}}\qubit{\#^2} \qubit{1^2}\qubit{\#} \qubit{\psi^{(code)}_{2,l}} \qubit{\#^2} \cdots \qubit{1^{b(l)}}\qubit{\#}\qubit{\psi^{(code)}_{b(l),l}}$ for any number $l\in\nat$, where $\#$ is a new symbol used as a separator.
Similarly to the proof of Lemma \ref{QFA-to-QTM}, $M$ works roughly as follows. On input $(x,1^{m(x)},h(|x|))$, $M$ first sweeps its advice tape to locate  $\qubit{\psi^{(code)}_{m(x),|x|}}$ and prepares the quantum state $\qubit{0^e}$ in the form of inner states.
When $N_{m(x),|x|}$ is in inner state $q$ scanning $\sigma$, $M$ executes a quantum circuit $C^{(m(x),|x|)}_{q,\sigma}$ described by a series of quantum gates specified inside $\qubit{\psi_{m(x),|x|}}$ and approximately generates a quantum state $\sum_{(p,d,\xi)}\alpha_{(q,\sigma,p,d,\xi)} \qubit{p,d,\xi}$ from $\qubit{0^{e}}$, where $\alpha_{q,\sigma,p,d,\xi}$ is an appropriate transition amplitude incurred by $M$.
The expected runtime of $M$ is bounded by $O(t(m(x),|x|) |h(|x|)| \log{r(m(x),|x|)})$, which is at most $O(b(|x|) s(m(x),|x|) t(m(x),|x|) \log{r(m(x),|x|)})$.
\end{proof}

Proposition \ref{character-2sPBQ}(1) directly follows from Lemmas \ref{poly-honest-bound}(1) and \ref{QTM-quantum-advice}.

\begin{proofof}{Proposition \ref{character-2sPBQ}(1)}
Take any parameterized decision problem $(L,m)$ and let $\LL=\{(L_n^{(+)},L_n^{(-)})\}_{n\in\nat}$ denote the family of promise decision problems induced from $(L,m)$.

(Only if--part)
Let us assume that $(L,m)$ is in $\para\pt\bql/\qpoly\cap \phsp$.
Take an advised QTM $M$ together with an $r(|x|)$-bounded quantum advice function $h$ for which $M$ with $h$ solves $(L,m)$ with bounded-error probability in expected $p(m(x),|x|)$ time using at most $\ell(m(x))$ space, where $p$ and $r$ are fixed polynomials and $\ell$ is a fixed logarithmic function.
Since $m$ is polynomially honest, there is an absolute constant $c>0$ such that,  for any $n$ and $x$, $x\in L_n^{(+)}\cup L_n^{(-)}$ implies $|x|\leq n^c+c$. This concludes that $\LL$ has a polynomial ceiling.

By Lemma \ref{QTM-quantum-advice}(1), we can convert the pair $(M,h)$ into a nonuniform family $\{N_{n,l}\}_{n,l\in\nat}$ of 2sqfa's with $O(r(l))\cdot 2^{O(\ell(n,l))} \subseteq (nl)^{O(1)}$ inner states together with a family $\{\qubit{\psi_{n,l}}\}_{n,l\in\nat}$ of quantum transition tables of encoded length $O(r(l)^9\log^2{p(n,l)})\cdot 2^{O(\ell(n,l))} \subseteq (nl)^{O(1)}$ since $2^{O(\ell(n,l))}\subseteq (nl)^{O(1)}$.
Since $|x|\leq n^c+c$ for any input $x\in L_n^{(+)}\cup L_n^{(-)}$, it suffices to consider the case of $l\leq n^c+c$.
Let us define $N'_n$ that works as follows.
We prepare a new quantum transition table $\qubit{\psi'_n} = \qubit{1^1}\qubit{\#}\qubit{\psi_{n,1}}\qubit{\#^2} \qubit{1^2}\qubit{\#} \qubit{\psi_{n,2}} \qubit{\#^2} \cdots
\qubit{1^{n^c+c}}\qubit{\#} \qubit{\psi_{n,n^c+c}}$.
On input $x$,  using $n^{O(1)}$ extra inner states, $N'_n$ first sets $l=|x|$ if $|x|\leq n^c+c$ and $l=n^c+c$ otherwise, and $N'_n$ then simulates $N_{n,l}$ on $x$ following the quantum transition table $\qubit{\psi_{n,l}}$ chosen from $\qubit{\psi'_n}$.
The expected runtime of $N'_n$ is bounded by $(n|x|)^{O(1)}$.
Finally,  we set $\NN'=\{N'_{n}\}_{n\in\nat}$ and $\Psi'=\{\qubit{\psi'_n}\}_{n\in\nat}$. Note that $\Psi'$ has encoded length of $n^{O(1)}$.

It thus follows that $M$ accepts (resp., rejects) $(x,h(|x|))$ with bounded-error probability iff $N'_{m(x)}$ accepts (resp., rejects) $x$ following $\qubit{\psi'_{m(x)}}$ with bounded-error probability. Since $M$ with $h$ solves $(L,m)$ with bounded-error probability, $\NN'$ solves $\LL$ following $\Psi'$ with bounded-error probability as well. This implies that $\LL$ is in $\pt\twosbq/\poly$.

(If--part)
For the converse implication, we assume that $\LL$ is in $\pt\twosbq/\poly$. Note that $|x|=n^{O(1)}$ holds for all $x\in L_n^{(+)}\cup L_n^{(-)}$. Take a family $\Psi = \{\qubit{\psi_n}\}_{n\in\nat}$ of quantum transition tables of encoded length at most $s(n)$ and a nonuniform family $\NN= \{N_n\}_{n\in\nat}$ of $r(n)$-state 2sqfa's that solves $\LL$ in expected $t(n,|x|)$ time with bounded-error probability by following  the transition tables embedded in  $\Psi$, where $r$, $s$, and $t$ are suitable polynomials.
Since $\{x\mid m(x)=n\} = L_n^{(+)}\cup L_n^{(-)}$, $m$ is polynomially bounded. Let us take a polynomial $b$ satisfying $m(x)\leq b(|x|)$ for all $x$.
Since $\LL$ is induced from $(L,m)$, Lemma \ref{poly-honest-bound}(1)   implies that $m$ is polynomially honest; thus, $(L,m)$ belongs to $\phsp$. By Lemma \ref{QTM-quantum-advice}(2), there exist an advised QTM $M$ and a polynomially-bounded quantum advice function $h$ such that $M$ with $h$ simulates $\NN$ with $\Psi$ in expected time $O(b(|x|)s(m(x))t(m(x),|x|) \log{r(m(x))}) \subseteq (|x|m(x))^{O(1)}$ using space $O(\log{r(m(x))}) \subseteq O(\log{m(x)})$. We therefore conclude that $(L,m)$ belongs to $\para\pt\bql/\qpoly$.
\end{proofof}

For the remaining proof of Proposition \ref{character-2sPBQ}(2), we need to claim the second crucial lemma, Lemma \ref{2sqfa-to-advice}, which can be shown in parallel to  Lemmas \ref{para-promise-QTM}--\ref{para-promise-2qfa}.

\begin{lemma}\label{2sqfa-to-advice}
Let $m$ denote a size parameter and let $b$, $r$, $s$, $t$, and $\ell$ be
arbitrary functions.
\renewcommand{\labelitemi}{$\circ$}
\begin{enumerate}\vs{-2}
  \setlength{\topsep}{-2mm}%
  \setlength{\itemsep}{1mm}%
  \setlength{\parskip}{0cm}%

\item Take an advised QTM $M$ and a quantum advice function $h$ with an advice alphabet $\Theta$. Assume that $h(l)$ produces a superposition of strings in $\Theta^{r(l)}$. There exist a nonuniform family $\NN= \{N_{n,l}\}_{n,l\in\nat}$ of 2sqfa's with $O(r(l))\cdot 2^{O(\ell(n,l))}$ states and a family $\Psi = \{\qubit{\psi_{n,l}}\}_{n,l\in\nat}$ of $\NN$'s quantum transition tables of encoded length $O(r(l)^9 \log^2t(n,l))\cdot 2^{O(\ell(n,l))}$ such that, for any string $x$, if $M$ accepts (resp., rejects) $(1^{n}\#x,h(|x|))$ with bounded-error probability in expected  $t(n,|x|)$ time using at most $\ell(n,|x|)$  space, then $N_{n,|x|}$  accepts (resp., rejects) $x$ following the quantum transition table $\qubit{\psi_{n,|x|}}$ with bounded-error probability in expected $O(t(n,|x|))$ time.

\item Let $\NN= \{N_{n,l}\}_{n,l\in\nat}$ be a nonuniform family of $r(n,l)$-state 2sqfa's and a family $\{\qubit{\psi_{n,l}}\}_{n,l\in\nat}$ of superpositions of $\NN$'s transition tables of encoded length $s(n,l)$. There exist an advised QTM $M$ and an  $O(ns(n,|x|))$-bounded quantum advice function $h$ such that, for any string $x$, if $N_{n,|x|}$ with the quantum transition table $\qubit{\psi_{n,|x|}}$ accepts (resp., rejects) $x$ with bounded-error probability in expected $t(n,|x|)$ time, then $M$ accepts (resp., rejects) $(1^{n}\#x,h(|x|))$ with bounded-error probability in expected $O(ns(n,|x|)t(n,|x|)\log{r(n,|x|)})$ time using $O(\log{r(n,|x|)})$ space.
\end{enumerate}
\end{lemma}

\begin{proof}
A basic idea of the proof of the lemma is a combination of the proofs of Lemma \ref{QTM-quantum-advice} and Lemmas \ref{para-promise-QTM}--\ref{para-promise-2qfa}.
Therefore, we give only the sketch of the proof.

(1) Take an advised QTM $M$ and a quantum advice function $h$ as in the premise of the lemma. Assume that $h(l)$ has the form $\sum_{a:|a|=r(l)}\alpha_a\qubit{a}$ over an advice alphabet $\Theta$.
For each advice string  $a\in\Theta^{r(l)}$,
we define $M_a$ from $M$ in a way similar to the proof of
Lemma \ref{QTM-quantum-advice}(1).
By a construction similar to the one used in the proof of Lemma \ref{para-promise-QTM}, we first build a family $\{N_{n,l,a}\}_{n,l,a}$ of 2sqfa's and a family $\{T_{n,l,a}\}_{n,l,a}$ of transition tables of encoded length $O(r(l)^9\log^2p(n,l))\cdot 2^{O(\ell(n,l))}$. From these $N_{n,l,a}$'s and $T_{n,l,a}$'s, we define  a quantum transition table $\qubit{\psi_{n,l}}$ as $\sum_{a:|a|=r(l)}\alpha_a\qubit{T_{n,l,a}}$ and a new 2sqfa $N'_{n,l}$ so that, for all $a\in\Theta^{r(|x|)}$, $N'_{n,l}$ simulates $N_{n,l,a}$ according to the quantum transition table $\qubit{T_{n,l,a}}$.

(2) Starting with a given family $\NN=\{N_{n,l}\}_{n,l\in\nat}$ of 2sqfa's and also a given family $\Psi = \{\qubit{\psi_{n,l}}\}_{n,l\in\nat}$ of quantum transition tables, we construct the desired advised QTM $M$ and the desired quantum advice function $h$ in a way similar to the proof of Lemma \ref{para-promise-2qfa} except that $h(l)$ holds the information on the superpositions of encoded transition tables of $N_{n,l}$ and that $M$ first retrieves a transition table from $h(|x|)$ and simulates $N_{n,|x|}$ on $x$.
\end{proof}

With the help of Lemmas \ref{poly-honest-bound}(2) and \ref{2sqfa-to-advice}, we intend to verify Proposition \ref{character-2sPBQ}(2).

\begin{proofof}{Proposition \ref{character-2sPBQ}(2)}
Let $\LL=\{(L_n^{(+)},L_n^{(-)})\}_{n\in\nat}$ be any $\dl$-good family of promise decision problems and let $(K,m)$ denote the parameterized decision problem induced from $\LL$.

(Only If--part) Assume that $(K,m)$ is in $\para\pt\bql/\qpoly\cap \phsp$.  There exist an advised QTM $M$ and a polynomially-bounded quantum advice function $h$ for which $M$ with $h$ solves $(K,m)$ with bounded-error probability in expected $t(m(x),|x|)$ time using at most $\ell(m(x))$ space, where $p$ is a fixed polynomial and $\ell$ is a fixed logarithmic function.
Take a suitable polynomial $r$ satisfying $|h(l)|\leq r(l)$ for all $l\in\nat$.
By Lemma \ref{2sqfa-to-advice}(1), we can take an appropriate nonuniform family $\NN= \{N_{n,l}\}_{n,l\in\nat}$ of expected-$O(t(n,l))$-time 2sqfa's having $O(r(l))\cdot 2^{O(\ell(n,l))}$ inner states for which $\NN$ with $\Psi$ solves $\LL$ with bounded-error probability following a certain family $\Psi = \{\qubit{\psi_{n,l}}\}_{n,l\in\nat}$ of  quantum transition tables of encoded length $O(r(l)^9\log^2{t(n,l)})\cdot 2^{O(\ell(n,l))}$.
Since $m$ is polynomially honest, there is a polynomial $s$ satisfying $|x|\leq s(m(x))$ for all $x$. This implies that $\LL$ has a polynomial ceiling.

We first define $\qubit{\psi'_n}$ as $\qubit{1^1}\qubit{\#} \qubit{\psi_{n,1}}\qubit{\#^2} \qubit{1^2}\qubit{\#} \qubit{\psi_{n,2}}\qubit{\#^2} \cdots \qubit{1^{s(n)}}\qubit{\#} \qubit{\psi_{n,s(n)}}$. Note that $\qubit{\psi'_n}$ has encoded length $n^{O(1)}$ by the choice of $r$ and $\ell$.
We then make $N'_n$ simulate $N_{n,|x|}$ as follows: on input $x$, $N'_n$ sets $l=|x|$ if $|x|\leq s(n)$ and $l=s(n)$ otherwise, and follows the quantum transition table $\qubit{\psi_{n,l}}$ chosen from $\qubit{\psi'_n}$.
It then follows that $N'_n$ uses $n^{O(1)}$ inner states and solves $(L_n^{(+)},L_n^{(-)})$ in expected $n^{O(1)}$ time.
Therefore, $\LL$ belongs to $\pt\twosbq/\poly$.

(If--part) On the contrary, we assume that $\LL$ is in $\pt\twosbq/\poly$. Lemma \ref{poly-honest-bound}(2) implies that $m$ is polynomially honest; thus, $(K,m)$ belongs to $\phsp$. Take a polynomial $b$ satisfying $m(x)\leq b(|x|)$ for all $x$.
We thus aim at verifying that $(K,m)$ belongs to  $\para\pt\bql/\poly$.
Take a nonuniform family $\NN= \{N_n\}_{n\in\nat}$ of $r(n)$-state 2sqfa's solving $\LL$ with bounded-error probability in expected $t(n,|x|)$ time following a certain family $\Psi = \{\qubit{\psi_n}\}_{n\in\nat}$ of  quantum transition tables of encoded length $s(n,l)$, where $r$, $s$, and $t$ are suitable polynomials.
Since $m(x)\leq b(|x|)$ for all $x$,
Lemma \ref{2sqfa-to-advice}(2) guarantees the existence of an advised QTM $M$ and an $O(b(|x|)s(m(x),|x|))$-bounded quantum advice function $h$ for which $M$ with $h$ simulates $\NN$ in expected time $O(b(|x|)s(m(x),|x|)t(m(x),|x|)\log{r(m(x))}) \subseteq (|x|m(x))^{O(1)}$  using space $O(\log{r(m(x))}) \subseteq O(\log{m(x)})$. We therefore conclude that $(K,m)$ indeed falls in $\para\pt\bql/\qpoly$.
\end{proofof}

\section{Challenging Open Questions}

\emph{Finite automata} are generally regarded as one of the simplest models of computation because their behaviors are simple enough to describe and easy to execute. Throughout this exposition, we have aimed at establishing bridges between such simple models and space-bounded advised computations by way of utilizing parameterized decision  problems.
Our effort in this exposition has earned a partial success in making a limited progress in the \emph{theory of nonuniform state complexity}, which was initiated by Berman and Lingas \cite{BL77} and Sakoda and Sipser \cite{SS78} and lately revitalized by Kapoutsis \cite{Kap09,Kap12,Kap14}, Kapoutsis and Pighizzini \cite{KP15}, and Yamakami \cite{Yam18}.

To direct fruitful future research, we intend to provide a short list of challenging open questions to the avid reader.

\renewcommand{\labelitemi}{$\circ$}
\begin{enumerate}\vs{-2}
  \setlength{\topsep}{-2mm}%
  \setlength{\itemsep}{1mm}%
  \setlength{\parskip}{0cm}%

\item Computational models of one-way finite automata are relatively easier to analyze than two-way head models; however, we have not determined all inclusion and separation relationships among the nonuniform state complexity classes appearing in Figure \ref{fig:relationship}. An important task is to complete this figure by establishing all missing relationships in the figure; for instance, prove or disprove that $\onen\nsubseteq\twod$ and $2^{\onebp}\nsubseteq \onenq$.

\item In computational complexity theory, many intriguing complexity classes have been introduced and extensively studied over the  years.  In comparison, the \emph{theory of nonuniform state complexity} is still in a cradle stage of its development. Develop that to a full fledged theory by, e.g., further introducing natural, intriguing state complexity classes and studying their structural properties. A simple example of such complexity classes is quantum interactive proof systems with quantum finite automata verifiers \cite{NY09}. Even for  structural properties of the existing classes, there still remain numerous open questions of whether or not, e.g., $\twon$ is closed under complementation.

\item Since finite automata are simple in structure, they may be suitable to be used for numerous applications outside of automata theory. In this exposition, we have shown only a direct application of nonuniform state complexity to the precise characterizations of several main-stream advised space complexity classes. We need to discover more intriguing applications of nonuniform state complexity in other fields of computer science.

\item The choice of transition amplitudes and transition probabilities seems to endow underlying quantum and probabilistic finite automata with quite different computational power. For example, we experience a great difficulty in showing that $\twoq=\co\twoq$ and $\twop = \co\twop$ because of the lack of our understanding of complex transition amplitudes and real transition probabilities. A much better understanding of complex and real numbers are definitely needed in further analyses of various computations. See also \cite[Section 5.3]{Yam03} for relevant subjects.

\item We have studied automata families of polynomially many inner states as well as exponentially many inner states. Since there is an enormous gap between them, it seems natural to study finite automata families of ``sub-exponentially many'' inner states. One of such  sub-exponential state complexities is $2^{\ell(n)}$ with $\ell(n) = \Theta(\log^k{n})$ for a fixed constant $k\geq2$. The nonuniform state complexity class defined by automata families of such state complexity will naturally fill the gap between polynomial state complexity  and exponential state complexity and will surely replenish our study of nonuniform state complexity.

\item Another important consideration is \emph{unary languages} (or \emph{tally sets}). In general, finite automata working on unary input strings are succinctly called \emph{unary automata}. With the use of such unary automata, from any nonuniform state complexity class $\CC$ (such as $\onebq$ and $\pt\twobq$), we can define $\CC/\mathrm{unary}$ by restricting all underlying finite automata to be unary automata.
    For those restrictive classes, we need to prove various inclusion and separation relationships. As for the (non)deterministic cases, refer to, e.g., \cite{Kap09,Kap12,Kap14,KP15}.

\item Our main concern in this exposition has been \emph{nonuniform}  families of various types of finite automata. As another direction of the current research, we suggest to study a ``uniform'' variant of those families. A family $\{M_n\}_{n\in\nat}$ of automata is \emph{uniform} if there exists a DTM $N$ that, on input $1^n$, produces an appropriate encoding $\pair{M_n}$ of $M_n$ on its write-once output tape. Along this line of study, there are only a few interesting results found in  \cite{BL77,Kap14,Yam18}. We need to explore this subject further in connection to, e.g., main-stream complexity classes.
\end{enumerate}

\let\oldbibliography\thebibliography
\renewcommand{\thebibliography}[1]{%
  \oldbibliography{#1}%
  \setlength{\itemsep}{-1pt}%
}
\bibliographystyle{alpha}

\end{document}